\documentclass[a4paper,10pt]{article}

\usepackage[longnamesfirst]{natbib}
\usepackage{enumerate}

\usepackage{amsfonts}
\usepackage{amsthm}
\usepackage{amsmath}
\usepackage{graphicx}

\usepackage{dcolumn}
\newcolumntype{d}{@{}D{.}{.}{3}@{}}
\newlength{\figurewidth}

\DeclareMathOperator{\conv}{conv}
\DeclareMathOperator{\dom}{dom}
\DeclareMathOperator{\successors}{succ}
\DeclareMathOperator{\cone}{cone}

\numberwithin{equation}{section}
\newtheorem{theorem}[equation]{Theorem}

\newtheorem{proposition}[equation]{Proposition}
\newtheorem{lemma}[equation]{Lemma}

\newtheorem{construction}[equation]{Construction}
\newtheorem{remark}[equation]{Remark}
\newtheorem{example}[equation]{Example}

\newcommand{\cc}{j}
\newcommand{\Pc}{\mathcal{P}_{\!j}}
\newcommand{\Pbarc}{\bar{\mathcal{P}}_{\!j}}

\newcommand{\Pbaruc}{\bar{\mathcal{P}}^\mathrm{d}_{\!j}}
\newcommand{\sigmac}{\sigma_{\!j}}

\newcommand{\st}{{\,|\,}}

\usepackage[colorlinks,linkcolor={blue},
 bookmarks={false},
 pdfpagelayout={SinglePage}
 ]{hyperref}

\title{American options with gradual exercise under proportional transaction costs}
\author{Alet Roux and Tomasz Zastawniak}

\begin{document}

\maketitle

\begin{abstract}
American options in a multi-asset market model with proportional transaction
costs are studied in the case when the holder of an option is able to exercise
it gradually at a so-called mixed (randomised) stopping time. The introduction of gradual exercise leads to tighter bounds on the option price when compared to the case studied in the existing literature, where the standard assumption is that the option can only be exercised instantly at an ordinary stopping time. Algorithmic
constructions for the bid and ask prices and the associated superhedging
strategies and optimal mixed stoping times for an American option with gradual
exercise are developed and implemented, and dual representations are established.
\end{abstract}

\section{Introduction}

This work on pricing American options under proportional transaction costs
goes back to the seminal discovery by \cite{Chalasani_Jha2001} that to
hedge against a buyer who can exercise the option at any (ordinary) stopping
time, the seller must in effect be protected against all mixed (randomised)
stopping times. This was followed by \cite{Bouchard_Temam2005}, who
established a non-constructive dual representation for the set of strategies
superhedging the seller's (though not the buyer's) position in an American
option under transaction costs. Efficient iterative algorithms for computing
the upper and lower hedging prices of the option, the hedging
strategies, optimal stopping times as well as dual representations for both
the seller and the buyer of an American option under transaction costs were
developed by \cite{Roux_Zastawniak2009} in a model with two assets, and
\cite{Roux_Zastawniak2011} in a multi-asset model. All these approaches
take it for granted that the buyer can only exercise the option instantly, at
an ordinary stopping time of his choosing.

By contrast, in the present paper we allow the buyer the flexibility to
exercise an American option gradually, rather than all at a single time
instance. Though it would be difficult in practice to exercise a fraction of
an option contract and to hold on to the reminder to exercise it later, the
holder of a large portfolio of options may well choose to exercise the individual
contracts on different dates if that proves beneficial. Does this ability to exercise gradually affect the pricing bounds, hedging strategies and optimal stopping times for the
buyer and/or seller? Perhaps surprisingly, the answer to this question is yes, it does in
the presence of transaction costs.

Gradual exercise turns out to be linked to another feature, referred to as
deferred solvency, which will also be studied here. If a temporary loss of liquidity
occurs in the market, as reflected by unusually large bid-ask spreads, agents
may become insolvent. Being allowed to defer closing their
positions until liquidity is restored might enable them to become solvent once
again. This gives more leeway when constructing hedging strategies than the
usual requirement that agents should remain solvent at all times.

\cite{Tien2012} was the first to explore the consequences of gradual exercise
and deferred solvency using a model with a single risky asset as a testing
ground. In the present paper these ideas are developed in a systematic manner
and extended to the much more general setting of the multi-asset market model
with transaction costs due to \cite{kabanov1999}; see also
\cite{kabanov_stricker2001b} and \cite{schachermayer2004}.

Pricing and hedging for the seller of an American option under
transaction costs is a convex optimisation problem irrespective of whether
instant or gradual exercise is permitted. However, this is not so for the buyer.
In this case one has to tackle a non-convex optimisation problem for options
that can only be exercised instantly. A
very interesting consequence of gradual exercise is that pricing and
hedging becomes a convex optimisation problem also for the buyer of an
American option, making it possible to deploy convex duality methods. The convexity of the problem also makes it much easier to implement the pricing and hedging algorithms numerically. We will make
use of this new opportunity in this paper.

The paper is organised as follows. Section~\ref{Sect-multi-curr-mod} recalls
the general setting of Kabanov's multi-asset model with transaction costs. In
Section~\ref{Sect:inst-versus-grad-exe} the hedging strategies for the buyer
and seller and the corresponding option prices under gradual exercise are
introduced and compared with the same notions under instant exercise. A
toy example is set up to demonstrate that it is easier to hedge an option and
that the bid-ask spread of the option prices can be narrower under gradual
exercise as compared to instant exercise. In Section~\ref{Sect:seller} the
seller's case is studied in detail. The notion of deferred solvency is first
discussed and linked in Proposition~\ref{prop:am:seller:immediate-ultimate} with the hedging problem for the seller of an
American option with gradual exercise. The
sets of seller's hedging portfolios are then constructed and related to the ask price of
the option under gradual exercise and to a construction of a seller's hedging
strategy realising the ask price; see
Theorem~\ref{prop:seller:Zau0=initial-endowments}. A dual representation of
the seller's price is established in Theorem~\ref{Thm:ask-price-representation}. The toy example is revisited to
illustrate the various constructions and results for the seller.
Section~\ref{Sect:buyer} is devoted to the buyer's case. Buyer's hedging
portfolios and strategies are constructed and used to compute the bid price
of the option; see Theorem~\ref{prop:2012-07-26:hedging-construct}. Finally,
the dual representation for the buyer is explored in Theorem~\ref{th:bu-buyer}. Once again, the toy example serves to illustrate the results. A numerical example with three assets can be found in Section~\ref{sec:num-example}. Some
conclusions and possible further developments and ramifications are touched
upon in Section~\ref{Sect:conclusions}. Technical information and proofs
are collected in the Appendix.

\section{Multi-currency model with proportional transaction costs}\label{Sect-multi-curr-mod}

Let $(\Omega,\mathcal{F},\mathbb{P};(\mathcal{F}_t)_{t=0}^T)$ be a filtered probability space. We assume that $\Omega$ is finite, $\mathcal{F}_0=\{\emptyset,\Omega\}$, $\mathcal{F}_T=\mathcal{F}=2^\Omega$ and $\mathbb{P}(\omega)>0$ for all $\omega\in\Omega$. For each $t=0,\ldots,T$ let $\Omega_t$ be the collection of atoms of $\mathcal{F}_t$, called the \emph{nodes} of the associated tree model. A node $\nu\in\Omega_{t+1}$ is said to be a \emph{successor} of a node $\mu\in\Omega_t$ if $\nu\subseteq\mu$. For each $t=0,\ldots,T-1$ we denote the collection of successors of any given node $\mu\in\Omega_t$ by $\successors \mu$.

For each $t=0,\ldots,T$ let $\mathcal{L}_t:=\mathcal{L}^0(\mathbb{R}^d;\mathcal{F}_t)$ be the collection of $\mathcal{F}_t$-measurable $\mathbb{R}^d$-valued random variables. We identify elements of $\mathcal{L}_t$ with functions on $\Omega_t$ whenever convenient.

We consider the discrete-time currency model introduced by \citet{kabanov1999} and developed further by \citet{kabanov_stricker2001b} and \citet{schachermayer2004} among others. The model contains $d$ assets or currencies. At each trading date $t=0,1,\ldots,T$ and for each $k,j=1,\ldots,d$ one unit of asset~$k$ can be obtained by exchanging $\pi^{jk}_t>0$ units of asset~$j$. We assume that the exchange rates $\pi^{jk}_t$ are $\mathcal{F}_t$-measurable and $\pi^{jj}_t=1$ for all $t$ and $j,k$.

We say that a portfolio $x\in\mathcal{L}_t$ is can be \emph{exchanged} into a portfolio $y\in\mathcal{L}_t$ at time~$t$ whenever there are $\mathcal{F}_t$-measurable random variables $\beta^{jk}\ge0$, $j,k=1,\ldots,d$ such that for all $k=1,\ldots,d$
\begin{equation}
y^k=x^k+\sum_{j=1}^d \beta^{jk} - \sum_{j=1}^d \beta^{kj}\pi^{kj}_t,
\end{equation}
where $\beta^{jk}$ represents the number of units of asset~$k$ received as a result of exchanging some units of asset~$j$.

The \emph{solvency cone} $\mathcal{K}_t\subseteq\mathcal{L}_t$ is the set of portfolios that are \emph{solvent} at time~$t$, i.e.\ the portfolios at time $t$ that can be exchanged into portfolios with non-negative holdings in all $d$ assets.
It is straightforward to show that $\mathcal{K}_t$ is the convex cone generated by the canonical basis $e^1,\ldots,e^d$ of $\mathbb{R}^d$ and the vectors $\pi^{jk}_te^j-e^k$ for $j,k=1,\ldots,d$, and so $\mathcal{K}_t$ is a polyhedral cone, hence closed. Note that $\mathcal{K}_t$ contains all the non-negative elements of $\mathcal{L}_t$.

A \emph{trading strategy} $y=(y_t)_{t=0}^{T+1}$ is a predictable $\mathbb{R}^d$-valued process with final value $y^{}_{T+1}=0$ and initial endowment $y_0\in\mathbb{R}^d$. For each $t>0$ the portfolio $y_t\in\mathcal{L}_{t-1}$ is held from time $t-1$ to time $t$. Let~$\Phi$ be the set of trading strategies.
We say that $y\in\Phi$ is a \emph{self-financing} strategy whenever $y_t-y_{t+1}\in\mathcal{K}_t$ for all $t=0,\ldots,T-1$.
Note that no implicitly assumed self-financing condition is included in the definition of~$\Phi$.

A trading strategy $y\in\Phi$ is an \emph{arbitrage opportunity} if it is self-financing, $y_0=0$ and there is a portfolio $x\in\mathcal{L}_T\setminus\{0\}$ with non-negative holdings in all $d$ assets such that $y^{}_T-x\in\mathcal{K}_T$. This notion of arbitrage was considered by \citet{schachermayer2004}, and its absence is formally different but equivalent to the weak no-arbitrage condition introduced by \citet{kabanov_stricker2001b}.

\begin{theorem}[\citet{kabanov_stricker2001b,schachermayer2004}]\label{th:2012-10-03:ftap}
 The model admits no arbitrage opportunity if and only if there exists a probability measure~$\mathbb{Q}$ equivalent to~$\mathbb{P}$ and an $\mathbb{R}^d$-valued $\mathbb{Q}$-martingale $S=(S_t)_{t=0}^T$ such that
 \begin{equation} \label{eq:th:2012-10-03:ftap}
  S_t\in\mathcal{K}_t^\ast\setminus\{0\} \text{ for all }t=0,\ldots,T,
 \end{equation}
where $\mathcal{K}^\ast_t$ is the polar of $-\mathcal{K}_t$; see {\upshape(\ref{eq:2012-09-20:Aast})} in the Appendix.
\end{theorem}

We denote by $\mathcal{P}$ the set of pairs $(\mathbb{Q},S)$ satisfying the conditions in Theorem~\ref{th:2012-10-03:ftap}, and by $\bar{\mathcal{P}}$ the set of pairs $(\mathbb{Q},S)$ satisfying the conditions in Theorem~\ref{th:2012-10-03:ftap} but with $\mathbb{Q}$  absolutely continuous with respect to (and not necessarily equivalent to)~$\mathbb{P}$. We assume for the remainder of this paper that the model admits no arbitrage opportunities, i.e.~$\mathcal{P}\neq\emptyset$.

\begin{remark}\upshape
In place of a pair $(\mathbb{Q},S)\in\mathcal{P}$ one can equivalently use the so-called
\emph{consistent price process} $S_{t}\mathbb{E}%
_{\mathbb{P}}(\frac{d\mathbb{Q}}{d\mathbb{P}}|\mathcal{F}_{t})$; see \cite{schachermayer2004}.
\end{remark}

We also define for any $\cc=1,\ldots,d$
\begin{align*}
 \Pc &:= \{(\mathbb{Q},S) \in\mathcal{P} \st S^\cc=1\}, &
 \Pbarc &:= \{(\mathbb{Q},S)\in\bar{\mathcal{P}} \st S^\cc=1\}.
\end{align*}
In the absence of arbitrage $\mathcal{K}^\ast_t$ is a non-empty compactly $\cc$-generated polyhedral cone for all $t$ \cite[Remark~2.2]{Roux_Zastawniak2011}, which means that $\Pc\neq\emptyset$. (For the definition of a compactly $\cc$-generated cone, see Appendix~\ref{Subsect:comp-gen-cones}.)

\section{Instant versus gradual exercise}\label{Sect:inst-versus-grad-exe}

The payoff of an American option in the model with $d$ underlying currencies is, in general, an $\mathbb{R}^d$-valued adapted process $\xi=(\xi_t)_{t=0}^T$.
The seller of the American option is obliged to deliver, and the buyer is entitled to receive the portfolio of currencies $\xi_\tau$ at a stopping time $\tau\in\mathcal{T}$ chosen by the buyer. Here $\mathcal{T}$ denotes the family of stopping times with values in $\{0,\ldots,T\}$.
This is the usual setup in which the option is exercised \emph{instantly} at a stopping time~$\tau\in\mathcal{T}$.

American options with the provision for instant exercise in the multi-currency model under proportional transaction costs have been studied by \citet{Bouchard_Temam2005}, who established a non-constructive characterisation of the superhedging strategies for the option seller only, and by \citet{Roux_Zastawniak2011}, who provided computationally efficient iterative constructions of the ask and bid option prices and the superhedging strategies for both the option seller and buyer.

In the present paper we relax the requirement that the option needs to be exercised instantly at a stopping time $\tau\in\mathcal{T}$. Instead, we allow the buyer to exercise \emph{gradually} at a mixed stopping time $\chi\in\mathcal{X}$. (For the definition of mixed stopping times, see Appendix~\ref{Sect:mixed-stop-times}.)

If the buyer chooses to exercise the option gradually according to a mixed stopping time $\chi\in\mathcal{X}$, then the seller of the American option will be obliged to deliver, and the buyer will be entitled to receive the fraction~$\chi_t$ of the portfolio of currencies~$\xi_t$ at each time $t=0,\ldots,T$.

The question then arises whether or not it would be more beneficial for the buyer to exercise the option gradually rather than instantly? What will be the optimal mixed stopping time $\chi\in\mathcal{X}$ for the buyer? How should the seller hedge against gradual exercise? Are the ask (seller's) and bid (buyer's) option prices and hedging strategies affected by gradual exercise as compared to instant exercise?

\subsection{Instant exercise}

In the case of instant exercise the seller of an American option~$\xi$ needs to hedge by means of a trading strategy~$y\in\Phi$ against all ordinary stopping times $\tau\in\mathcal{T}$ chosen by the buyer. The trading strategy~$y$ needs to be self-financing up to time~$\tau$ and to allow the seller to remain solvent on delivering the portfolio $\xi_\tau$ at time~$\tau$, for any~$\tau\in\mathcal{T}$. Hence the family of seller's superhedging strategies is defined as
\[
 \Phi^\mathrm{a}(\xi):=\{y\in\Phi \st \forall \tau\in\mathcal{T}:  y_t-y_{t+1}\in\mathcal{K}_t\text{ for } t=0,\ldots,\tau-1,
 y_\tau-\xi_\tau\in\mathcal{K}_\tau\},
\]
and the \emph{ask price} (\emph{seller's price}) of the option in currency $\cc=1,\ldots,d$ is
\[
 \pi^\mathrm{a}_\cc(\xi):=\inf\{x\in\mathbb{R} \st \exists y\in\Phi^\mathrm{a}(\xi):xe^\cc=y_0\}.
\]
This is the smallest amount in currency~$\cc$ needed to superhedge a short position in~$\xi$.

On the other hand, the buyer of an American option~$\xi$ can select both a stopping time~$\tau\in\mathcal{T}$ and a trading strategy~$y\in\Phi$. The trading strategy~$y$ needs to be self-financing up to time~$\tau$ and to allow the buyer to remain solvent on receiving the portfolio $\xi_\tau$ at time~$\tau$. Thus, the family of buyer's superhedging strategies is defined as
\[
 \Phi^\mathrm{b}(\xi):=\{(y,\tau)\in\Phi\times\mathcal{T} \st
 y_t-y_{t+1}\in\mathcal{K}_t\text{ for }t=1,\ldots,\tau-1,
y_\tau+\xi_\tau\in\mathcal{K}_\tau\},
\]
and the \emph{bid price} (\emph{buyer's price}) of the option in currency $\cc=1,\ldots,d$ is
\[
 \pi^\mathrm{b}_\cc(\xi):=\sup\{-x\in\mathbb{R} \st \exists (y,\tau)\in\Phi^\mathrm{b}(\xi): xe^\cc=y_0\}.
\]
This is the largest amount in currency~$\cc$ that the buyer can raise using the option~$\xi$ as surety.

For American options with instant exercise, iterative constructions of the ask and bid option prices $\pi^\mathrm{a}_\cc(\xi)$ and $\pi^\mathrm{b}_\cc(\xi)$ and the corresponding seller's and buyer's superhedging strategies from $\Phi^\mathrm{a}(\xi)$ and $\Phi^\mathrm{b}(\xi)$ were established by \citet{Roux_Zastawniak2011}.

\subsection{Gradual exercise}\label{Sect:grad-exe}

When the buyer is allowed to exercise gradually, the seller needs to follow a suitable trading strategy to hedge his exposure. Since the seller can react to the buyer's actions, this strategy may in general depend on the mixed stopping time $\chi\in\mathcal{X}$ followed by the buyer, and will be denoted by $Y^\chi\in\Phi$. In other words, we consider a function $Y:\mathcal{X}\to\Phi$.

At each time $t$ the seller will be holding a portfolio $Y^\chi_t$ and will be obliged to deliver a fraction $\chi_t$ of the payoff $\xi_t$. He can then rebalance the remaining portfolio $Y^\chi_t-\chi_t\xi_t$ into $Y^\chi_{t+1}$ in a self-financing manner, so that
\begin{equation}\label{eq:seller-self-fin-superhedge}
Y^\chi_t-\chi_t\xi_t-Y^\chi_{t+1}\in\mathcal{K}_t
\text{ for each }t=0,\ldots,T.
\end{equation}The self-financing and superhedging conditions have merged into one. We call (\ref{eq:seller-self-fin-superhedge}) the \emph{rebalancing} condition.

When creating the portfolio $Y^\chi_t$ at time~$t-1$, the seller can only use information available at that time. This includes $\chi_0,\ldots,\chi_{t-1}$, but the seller has no way of knowing the future values $\chi_t,\ldots,\chi_T$ that will be chosen by the buyer. The trading strategies $Y^\chi\in\Phi$ that can be adopted by the seller are therefore restricted to those satisfying the \emph{non-anticipation} condition
\begin{equation}\label{eq:non-anticipate}
\begin{array}{c}
\forall\chi,\chi'\in\mathcal{X}\ \forall t=0,\ldots,T\ \forall\omega\in\Omega:\\
\chi_{s}^{}(\omega)=\chi'_{s}(\omega)\text{ for each }s=0,\ldots,t-1\Longrightarrow
Y_{t}^{\chi}(\omega)=Y_{t}^{\chi'}(\omega).
\end{array}
\end{equation}
In particular, the initial endowment $Y^\chi_0$ of the trading strategy $Y^\chi$ is the same for all $\chi\in\mathcal{X}$. We denote this common value by~$Y_0$.

We define the family of seller's superhedging strategies against gradual exercise by
\[
\Phi^\mathrm{ag}(\xi):=\{Y:\mathcal{X}\to\Phi \st Y\text{ satisfies (\ref{eq:seller-self-fin-superhedge}) and (\ref{eq:non-anticipate})}\}
\]
and the corresponding \emph{ask price} (\emph{seller's price}) of the option in currency $j=1,\ldots,d$ by
\begin{equation}\label{eq:20130728:def-pi-ag}
\pi^\mathrm{ag}_\cc(\xi):=\inf\{x\in\mathbb{R}\st \exists Y\in\Phi^\mathrm{ag}(\xi): xe^\cc=Y_0\}.
\end{equation}
This is the smallest amount in currency~$j$ that the seller needs to superhedge a short position in the American option~$\xi$ when the buyer is allowed to exercise gradually.

On the other hand, the buyer is able to select both a mixed stopping time $\chi\in\mathcal{X}$ and a trading strategy $y\in\Phi$, and will be taking delivery of a fraction~$\chi_t$ of the payoff~$\xi_t$ at each time~$t$. Because the choice of the mixed stopping time~$\chi$ is up to the buyer, the trading strategy $y$  needs to be good just for the one chosen stopping time, and does not need to be considered as a function of~$\chi$, in contrast to the seller's case. The \emph{rebalancing} condition
\begin{equation}\label{eq:buyer-self-fin-superhedge}
y_t+\chi_t\xi_t-y_{t+1}\in\mathcal{K}_t
\text{ for each }t=0,\ldots,T
\end{equation}
needs to be satisfied.

Hence, the family of superhedging strategies for the buyer of an American option~$\xi$ with gradual exercise is defined as
\[
\Phi^\mathrm{bg}(\xi):=\{(y,\chi)\in\Phi\times\mathcal{X} \st (y,\chi)\text{ satisfies (\ref{eq:buyer-self-fin-superhedge})}\},
\]
and the corresponding \emph{bid price} (\emph{buyer's price}) of the option in currency $j=1,\ldots,d$ is
\begin{equation}\label{eq:buyer-bid-price-gradual}
\pi^\mathrm{bg}_\cc(\xi):=\sup\{-x\in\mathbb{R}\st \exists (y,\chi)\in\Phi^\mathrm{bg}(\xi): xe^\cc=y_0\}.
\end{equation}
This is the largest amount in currency~$j$ that can be raised using the option as surety by a buyer who is able to exercise gradually.

\begin{example}\label{Exl:new}\upshape
We consider a toy example with two assets, a foreign
currency (asset 1) and domestic currency (asset 2) in a two-step binomial tree
model with the following bid/ask foreign currency prices $S_{t}^{\mathrm{b}%
}\leq S_{t}^{\mathrm{a}}$ in each of the four scenarios in $\Omega=\left\{
\omega_{1},\omega_{2},\omega_{3},\omega_{4}\right\}  $:%
\[%
\begin{tabular}
[c]{|c|cc|cc|cc|}\hline
$\rule[-0.2cm]{0pt}{0.6cm}$ & $S_{0}^{\mathrm{b}}$ & $S_{0}^{\mathrm{a}}$ &
$S_{1}^{\mathrm{b}}$ & $S_{1}^{\mathrm{a}}$ & $S_{2}^{\mathrm{b}}$ &
$S_{2}^{\mathrm{a}}$\\\hline
$\omega_{1}$ & $5$ & $5$ & $3$ & $9$ & $4$ & $8$\\\cline{6-7}%
$\omega_{2}$ & $5$ & $5$ & $3$ & $9$ & $4$ & $4$\\\cline{4-7}%
$\omega_{3}$ & $5$ & $5$ & $2$ & $2$ & $3$ & $3$\\\cline{6-7}%
$\omega_{4}$ & $5$ & $5$ & $2$ & $2$ & $1$ & $1$\\\hline
\end{tabular}
\]
Note there are only two nodes with a non-trivial bid/ask spread, namely the `up' node $\mathrm{U}=\left\{
\omega_{1},\omega_{2}\right\}  $ and the `up-up' node $\mathrm{UU}%
=\{\omega_{1}\}$. The corresponding
exchange rates are
\[
\left[
\begin{array}
[c]{cc}%
\pi_{t}^{11} & \pi_{t}^{12}\\
\pi_{t}^{21} & \pi_{t}^{22}%
\end{array}
\right]  =\left[
\begin{array}
[c]{cc}%
1 & 1/S_{t}^{\mathrm{b}}\\
S_{t}^{\mathrm{a}} & 1
\end{array}
\right]  .
\]
In this model we consider an American option with the following payoff process
$\xi_{t}=\left(  \xi_{t}^{1},\xi_{t}^{2}\right)  $:%
\[%
\begin{tabular}
[c]{|c|c|c|c|}\hline
$\rule[-0.15cm]{0pt}{0.5cm}$ & $\xi_{0}$ & $\xi_{1}$ & $\xi_{2}$\\\hline
$\omega_{1}$ & $\left(  0,0\right)  $ & $\left(  0,4\right)  $ & $\left(
2,-8\right)  $\\\cline{4-4}%
$\omega_{2}$ & $\left(  0,0\right)  $ & $\left(  0,4\right)  $ & $\left(
0,0\right)  $\\\cline{3-4}%
$\omega_{3}$ & $\left(  0,0\right)  $ & $\left(  0,0\right)  $ & $\left(
0,0\right)  $\\\cline{4-4}%
$\omega_{4}$ & $\left(  0,0\right)  $ & $\left(  0,0\right)  $ & $\left(
0,0\right)  $\\\hline
\end{tabular}
\]
In the case when the option can only be exercised instantly, using the algorithms of \cite{Roux_Zastawniak2011} we can compute the bid and
ask prices of the option in the domestic currency to be%
\[
\pi_{2}^{\mathrm{b}}(\xi)=2,\quad\pi_{2}^{\mathrm{a}}(\xi)=\frac{28}{5}.
\]
Now consider $Y:\mathcal{X}\rightarrow\Phi$ given by%
\[%
\begin{tabular}
[c]{|c|c|c|c|}\hline
$\rule[-0.2cm]{0pt}{0.6cm}$ & $Y_{0}^{\chi}$ & $Y_{1}^{\chi}$ & $Y_{2}^{\chi}%
$\\\hline
$\omega_{1}\rule[-0.2cm]{0pt}{0.6cm}$ & $\left(  0,5\right)  $ & $\left(
1,0\right)  $ & $\left(  1,-4\chi_{1}^{\omega_{1}}\right)  $\\
$\omega_{2}\rule[-0.2cm]{0pt}{0.6cm}$ & $\left(  0,5\right)  $ & $\left(
1,0\right)  $ & $\left(  1,-4\chi_{1}^{\omega_{2}}\right)  $\\\cline{4-4}%
$\omega_{3}\rule[-0.2cm]{0pt}{0.6cm}$ & $\left(  0,5\right)  $ & $\left(
1,0\right)  $ & $\left(  0,0\right)  $\\
$\omega_{4}\rule[-0.2cm]{0pt}{0.6cm}$ & $\left(  0,5\right)  $ & $\left(
1,0\right)  $ & $\left(  0,0\right)  $\\\hline
\end{tabular}
\]
for any $\chi\in\mathcal{X}$. Also consider $y\in\Phi$ and $\chi\in
\mathcal{X}$ such that%
\[%
\begin{tabular}
[c]{|c|c|c|c|c|c|c|}\hline
$\rule[-0.15cm]{0pt}{0.5cm}$ & $y_{0}$ & $y_{1}$ & $y_{2}$ & $\chi_{0}$ &
$\chi_{1}$ & $\chi_{2}$\\\hline
$\omega_{1}\rule[-0.2cm]{0pt}{0.6cm}$ & $\left(  0,-3\right)  $ & $\left(
-1,2\right)  $ & $\left(  -1,4\right)  $ & $0$ & $\frac{1}{2}$ & $\frac{1}{2}%
$\\
$\omega_{2}\rule[-0.2cm]{0pt}{0.6cm}$ & $\left(  0,-3\right)  $ & $\left(
-1,2\right)  $ & $\left(  -1,4\right)  $ & $0$ & $\frac{1}{2}$ & $\frac{1}{2}%
$\\\cline{4-4}\cline{6-7}%
$\omega_{3}\rule[-0.2cm]{0pt}{0.6cm}$ & $\left(  0,-3\right)  $ & $\left(
-1,2\right)  $ & $\left(  0,0\right)  $ & $0$ & $0$ & $1$\\
$\omega_{4}\rule[-0.2cm]{0pt}{0.6cm}$ & $\left(  0,-3\right)  $ & $\left(
-1,2\right)  $ & $\left(  0,0\right)  $ & $0$ & $0$ & $1$\\\hline
\end{tabular}
\]
We can verify that $Y\in\Phi^{\mathrm{ag}}(\xi)$ and $\left(  y,\chi\right)
\in\Phi^{\mathrm{bg}}(\xi)$. The existence of these strategies means that%
\[
\pi_{2}^{\mathrm{b}}(\xi)=2<3\leq\pi_{2}^{\mathrm{bg}}(\xi),\quad\pi
_{2}^{\mathrm{ag}}(\xi)\leq5<\frac{28}{5}=\pi_{2}^{\mathrm{a}}(\xi).
\]
\end{example}

This example demonstrates that the seller's and buyer's prices $\pi^\mathrm{ag}_\cc(\xi),\pi^\mathrm{bg}_\cc(\xi)$ under gradual exercise may differ from their respective counterparts $\pi^\mathrm{a}_\cc(\xi),\pi^\mathrm{b}_\cc(\xi)$ under instant exercise.
It demonstrates the need to revisit and investigate the pricing and superhedging results in the case when the instant exercise provision is relaxed and replaced by gradual exercise.

\section{Pricing and superhedging for the seller under gradual exercise}\label{Sect:seller}

We have seen in Example~\ref{Exl:new} that the seller's price $\pi_\cc^\mathrm{a}(\xi)$ may be higher than $\pi_\cc^\mathrm{ag}(\xi)$. The reason is that an option seller who follows a hedging strategy $y\in\Phi^\mathrm{a}(\xi)$ is required to be instantly solvent upon delivering the payoff at the stopping time $\tau\in\mathcal{T}$ when the buyer has chosen to exercise the option. Meanwhile, a seller who follows a strategy $Y\in\Phi^\mathrm{ag}(\xi)$ will be able to continue rebalancing the strategy up to the time horizon~$T$ as long as a solvent position can be reached eventually. Being able to defer solvency in this fashion allows more flexibility for the seller, resulting in a lower seller's price.

On the other hand, it might appear that a seller who hedges against gradual exercise (against mixed stopping times) would have a harder task to accomplish than someone who only needs to hedge against instant exercise (ordinary stopping times). However, this turns out not to be a factor affecting the seller's price, as we shall see in Proposition~\ref{prop:am:seller:immediate-ultimate}.

\subsection{Deferred solvency}

These considerations indicate that the notion of solvency needs to be relaxed.

We say that a portfolio $z\in\mathcal{L}_t$ satisfies the \emph{deferred solvency} condition at time~$t$ if it can be exchanged into a solvent portfolio by time $T$ without any additional investment, i.e.\ if there is a sequence $y_{t+1},\ldots,y_{T+1}$ such that $y_s\in\mathcal{L}_{s-1}$ for all $s=t+1\ldots,T$ and
\begin{align*}
z-y_{t+1}&\in\mathcal{K}_t, & y_s - y_{s+1}&\in\mathcal{K}_s\text{ for all }{s=t+1,\ldots,T}, & y_{T+1}=0.
\end{align*}
We call such a sequence $y_{t+1},\ldots,y_{T+1}$ a \emph{liquidation strategy} starting from~$z$ at time~$t$.

The set of portfolios satisfying the deferred solvency condition at time~$t$ is a cone. We call it the \emph{deferred solvency cone} and denote by $\mathcal{Q}_t$.

\begin{example}
\upshape In Example~\ref{Exl:new} the portfolio with~$8$ in the domestic
currency and~$-1$ in the foreign currency is insolvent at the `up' node
$\mathrm{U}=\{\omega_{1},\omega_{2}\}$ at time~$1$, that is, $(-1,8)\notin
\mathcal{K}_{1}^{\mathrm{U}}$. It does, however, satisfy the deferred solvency
condition at that node, i.e.\ $(-1,8)\in\mathcal{Q}_{1}^{\mathrm{U}}$. The large
bid-ask spread $[S_{1}^{\mathrm{bU}},S_{1}^{\mathrm{aU}}]=[3,9]$ at
node~$\mathrm{U}$ indicates a temporary loss of liquidity. Although the
portfolio is insolvent at that node, waiting until the market recovers from the loss of liquidity can restore solvency. The liquidation strategy is to hold the portfolio
until time~$2$ and to buy the foreign currency then.
\end{example}

The following result shows that the deferred solvency cones $\mathcal{Q}_t$ can be regarded as the sets of time~$t$ superhedging portfolios for the seller of a European option with expiry time~$T$ and zero payoff; see \cite{Roux_Zastawniak2011}.

\begin{proposition}\label{prop:constr:ultimate-solvent}
The deferred solvency cones can be constructed by backward induction as follows:
 \begin{align}
  \mathcal{Q}_T &= \mathcal{K}_T, \label{eq:QT=KT}\\
  \mathcal{Q}_t &= \mathcal{Q}_{t+1}\cap\mathcal{L}_t + \mathcal{K}_t \text{ for all }t=0,\ldots,T-1. \label{eq:Qt-recursive}
 \end{align}
\end{proposition}

The proof of Proposition~\ref{prop:constr:ultimate-solvent} can be found in Appendix~\ref{Sect:Appendix:defer-solv}.

From~(\ref{eq:Qt-recursive}) we can see that for any~$t=0,\ldots,T-1$ and for any $\mu\in\Omega_t$
\begin{equation}\label{eq:2013-07-12-Q-mu-t}
\mathcal{Q}^\mu_t=\bigcap_{\nu\in\successors\mu}\mathcal{Q}^\nu_{t+1}+\mathcal{K}_t^\mu.
\end{equation}
By backward induction, $\mathcal{Q}^\mu_t$ is given as an intersection and algebraic sum of a finite number of polyhedral cones, so it is a polyhedral cone. This also means the solvency cones can readily be computed using standard operations on polyhedral convex sets.

The next result shows that Theorem~\ref{th:2012-10-03:ftap} can be formulated equivalently in terms of the deferred solvency cones~$\mathcal{Q}_t$ instead of the solvency cones~$\mathcal{K}_t$.

\begin{proposition}\label{prop:mart-ito-Q}
If $\mathbb{Q}$ is a probability measure and $S=(S_t)_{t=0}^T$ is an $\mathbb{R}^d$-valued $\mathbb{Q}$-martingale, then $S$ satisfies \eqref{eq:th:2012-10-03:ftap} if and only if
 \begin{equation}\label{eq:2012-10-03:Qastt-not-empty}
  S_t\in\mathcal{Q}_t^\ast\setminus\{0\}\text{ for all }t=0,\ldots,T,
 \end{equation}
where $\mathcal{Q}^\ast_t$ is the polar of $-\mathcal{Q}_t$.
\end{proposition}

The proof of Proposition~\ref{prop:mart-ito-Q} is in Appendix~\ref{Sect:Appendix:defer-solv}.

\subsection{Construction of seller's price and superhedging strategy}

We extend the family $\Phi^\mathrm{a}(\xi)$ of seller's superhedging strategies by allowing for deferred solvency:
\[
 \Phi^\mathrm{ad}(\xi):=\{y\in\Phi \st \forall \tau\in\mathcal{T}: y_t-y_{t+1}\in\mathcal{Q}_t\text{ for }t=1,\ldots,\tau-1,
 y_\tau-\xi_\tau\in\mathcal{Q}_\tau\},
\]

The following proposition shows that the set of initial endowments that allow the seller to hedge against gradual exercise is the same as that allowing to hedge against instant exercise with deferred solvency.

\begin{proposition}\label{prop:am:seller:immediate-ultimate}
 For any American option $\xi$
 \[
  \{x\in\mathbb{R}^d\st \exists Y\in\Phi^\mathrm{ag}(\xi):x=Y_0\}
  =\{x\in\mathbb{R}^d\st \exists y\in\Phi^\mathrm{ad}(\xi):x=y_0\}.
 \]
\end{proposition}

For the proof of Proposition~\ref{prop:am:seller:immediate-ultimate}, see Appendix~\ref{Sect:Appendix:technical:seller}.

We now present an iterative construction of the set of initial endowments that allow superhedging for the seller under deferred solvency. By Proposition~\ref{prop:am:seller:immediate-ultimate}, this also gives the set of initial endowments that allow superhedging for the seller under gradual exercise.

\begin{construction}\label{const:am:seller}
 Construct adapted sequences $\mathcal{U}^\mathrm{ad}_t$, $\mathcal{V}^\mathrm{ad}_t$, $\mathcal{W}^\mathrm{ad}_t$, $\mathcal{Z}^\mathrm{ad}_t$ for $t=0,\ldots,T$ by
 \begin{align}
  \mathcal{U}^\mathrm{ad}_t&:=\xi_t+\mathcal{Q}_t \text{ for all }t=0,\ldots,T,\label{eq:seller:mathcal-Ut}\\
  \mathcal{Z}^\mathrm{ad}_T&:=
  \mathcal{V}^\mathrm{ad}_T:=
  \mathcal{W}^\mathrm{ad}_T:=
  \mathcal{U}^\mathrm{ad}_T,\nonumber\\
  \intertext{and by backward induction on all $t=0,\ldots,T-1$}
  \mathcal{W}^\mathrm{ad}_t&:=\mathcal{Z}^\mathrm{ad}_{t+1}\cap\mathcal{L}_t,\label{eq:seller:mathcal-Wt}\\
  \mathcal{V}^\mathrm{ad}_t&:=\mathcal{W}^\mathrm{ad}_t+\mathcal{Q}_t,\label{eq:seller:mathcal-Vt}\\
  \mathcal{Z}^\mathrm{ad}_t&:=\mathcal{U}^\mathrm{ad}_t\cap\mathcal{V}^\mathrm{ad}_t.\label{eq:seller:mathcal-Zt}
 \end{align}
\end{construction}

It follows by backward induction that the sets
\[
\mathcal{W}^\mathrm{ad\mu}_t=\bigcap_{\nu\in\successors\mu}\mathcal{Z}^\mathrm{ad\nu}_{t+1},
\quad\quad
\mathcal{V}^\mathrm{ad\mu}_t=\mathcal{W}^\mathrm{ad\mu}+\mathcal{Q}^\mu_t,
\quad\quad
\mathcal{Z}^\mathrm{ad\mu}_t=\mathcal{U}^\mathrm{ad\mu}_t\cap\mathcal{V}^\mathrm{ad\mu}_t
\]
are convex and polyhedral for each $t=0,\ldots,T-1$ and $\mu\in\Omega_t$ because the algebraic sum and the intersection of a finite number of convex polyhedral sets are convex and polyhedral, and
\[
\mathcal{Z}^{\mathrm{ad}\mu}_T=\mathcal{V}^{\mathrm{ad}\mu}_T
=\mathcal{W}^{\mathrm{ad}\mu}_T=\mathcal{U}^{\mathrm{ad}\mu}_T=\xi^\mu_T+\mathcal{Q}^\mu_T
\]
are convex polyhedral sets for each $\mu\in\Omega_T$.
Moreover, $\mathcal{U}^\mathrm{ad}_t$, $\mathcal{V}^\mathrm{ad}_t$, $\mathcal{W}^\mathrm{ad}_t$, $\mathcal{Z}^\mathrm{ad}_t$ are non-empty for each $t=0,\ldots,T$ because the portfolio $(z,\ldots,z)\in\mathbb{R}^d$ belongs to all of them when $z\in\mathbb{R}$ is large enough.

\begin{theorem}\label{prop:seller:Zau0=initial-endowments}
  The set of initial endowments that superhedge the seller's position in the American option~$\xi$ under gradual exercise is equal to
 \begin{equation}\label{eq:20130727-Zad0}
  \mathcal{Z}^\mathrm{ad}_0
  = \{x\in\mathbb{R}^d\st \exists Y\in\Phi^\mathrm{ag}(\xi):x=Y_0\},
 \end{equation}
and the ask (seller's) price of the option in currency~$j=1,\ldots,d$ can be computed as
\[
 \pi^\mathrm{ag}_\cc(\xi)
 =\min\{x\in\mathbb{R}\st xe^\cc\in\mathcal{Z}^\mathrm{ad}_0\}.
\]
Moreover, a strategy $Y\in\Phi^\mathrm{ag}(\xi)$ can be constructed such that
\[
\pi^\mathrm{ag}_\cc(\xi)e^\cc = Y_0.
\]
\end{theorem}

The proof of Theorem~\ref{prop:seller:Zau0=initial-endowments} can be found in Appendix~\ref{Sect:Appendix:technical:seller}.

We can conclude that the set of initial endowments~$\mathcal{Z}^\mathrm{ad}_0$ superhedging the seller's position, the option ask price~$\pi^\mathrm{ag}_\cc(\xi)$, and a superhedging strategy~$Y$ realising the ask price can be computed by means of standard operations on convex polyhedral sets.

\begin{example}
\upshape Working within the setting of Example~\ref{Exl:new}, we can now apply
the constructions described in the current section to compute the sets
$\mathcal{Z}_{t}^{\mathrm{ad}}$ of superhedging portfolios for the seller.
These are sets of portfolios $(x^{1},x^{2})\in\mathbb{R}^{2}$ satisfying the
inequalities%
\[%
\begin{tabular}
[c]{|c|c|c|c|}\hline
$\rule[-0.2cm]{0pt}{0.6cm}$ & $\mathcal{Z}_{0}^{\mathrm{ad}}$ & $\mathcal{Z}%
_{1}^{\mathrm{ad}}$ & $\mathcal{Z}_{2}^{\mathrm{ad}}$\\\hline
$\omega_{1}\rule[-0.4cm]{0cm}{1.05cm}$ & $5x^{1}+x^{2}\geq5$ & $%
\begin{array}
[c]{l}%
8x^{1}+x^{2}\geq8\\
4x^{1}+x^{2}\geq0
\end{array}
$ & $%
\begin{array}
[c]{l}%
8x^{1}+x^{2}\geq8\\
4x^{1}+x^{2}\geq0
\end{array}
$\\\cline{4-4}%
$\omega_{2}\rule[-0.4cm]{0cm}{1.05cm}$ & $5x^{1}+x^{2}\geq5$ & $%
\begin{array}
[c]{l}%
8x^{1}+x^{2}\geq8\\
4x^{1}+x^{2}\geq0
\end{array}
$ & $4x^{1}+x^{2}\geq0$\\\cline{3-4}%
$\omega_{3}\rule[-0.4cm]{0cm}{1.05cm}$ & $5x^{1}+x^{2}\geq5$ & $2x^{1}%
+x^{2}\geq0$ & $3x^{1}+x^{2}\geq0$\\\cline{4-4}%
$\omega_{4}\rule[-0.4cm]{0cm}{1.05cm}$ & $5x^{1}+x^{2}\geq5$ & $2x^{1}%
+x^{2}\geq0$ & $x^{1}+x^{2}\geq0$\\\hline
\end{tabular}
\]
From $\mathcal{Z}_{0}^{\mathrm{ad}}$ we obtain the ask price
\[
\pi_{2}^{\mathrm{ag}}(\xi)=\min\{x\in\mathbb{R}\st(0,x)\in\mathcal{Z}%
_{0}^{\mathrm{ad}}\}=5.
\]
We can also construct a superhedging strategy $Y\in\Phi^{\mathrm{ag}}(\xi)$
such that
\[
(0,\pi_{2}^{\mathrm{ag}}(\xi))=Y_{0}=(0,5).
\]
It is the strategy~$Y$ specified in Example~\ref{Exl:new}.
\end{example}

\subsection{Dual representation of seller's price}

A dual representation of the seller's price $\pi^\mathrm{ag}_\cc(\xi)$ can be obtained with the aid of the support function $Z^\mathrm{ad}_0$ of~$-\mathcal{Z}^\mathrm{ad}_0$. For the definition of the support function of a convex set, see Appendix~\ref{Subsect:conv-anal}. More generally, let $U^\mathrm{ad}_t$, $V^\mathrm{ad}_t$, $W^\mathrm{ad}_t$,~$Z^\mathrm{ad}_t$ be the support functions of the sets $-\mathcal{U}^\mathrm{ad}_t$, $-\mathcal{V}^\mathrm{ad}_t$, $-\mathcal{W}^\mathrm{ad}_t$, $-
\mathcal{Z}^\mathrm{ad}_t$ of Construction~\ref{const:am:seller}.
The functions $U^\mathrm{ad}_t$, $V^\mathrm{ad}_t$ $W^\mathrm{ad}_t$, $Z^\mathrm{ad}_t$ are polyhedral \cite[Corollary~19.2.1]{rockafellar1996}, hence continuous.
Proposition~\ref{prop:seller:dual} in Appendix~\ref{Sect:Appendix:technical:seller} lists a number of properties of support functions, which will prove useful in what follows.

\begin{proposition}\label{Prop:20130727:pi-ag-dual}
The seller's price of an American option~$\xi$ with gradual exercise can be written as
\[
 \pi^\mathrm{ag}_\cc(\xi)
 =\max\{-Z^\mathrm{ad}_0(s)\st s\in\sigmac(\mathbb{R}^d)\}
 =\mathbb{E}_\mathbb{Q}((\xi\cdot S)_\chi)
\]
for some mixed stopping time $\chi\in\mathcal{X}$, a probability measure $\mathbb{Q}$ and an $\mathbb{R}^d$-valued adapted process~$S$ such that
\begin{equation}\label{eq:chi-approx-u}
S_t\in\mathcal{Q}^\ast_t\setminus\{0\}\ \ \text{and}\ \,\, \mathbb{E}_\mathbb{Q}(S^{\chi\ast}_{t+1}|\mathcal{F}_t)\in\mathcal{Q}^\ast_t
\quad\text{for each }t=0,\ldots,T,
\end{equation}
and $S^j_t=1$ for all $t=0,\ldots,T$. Such $\chi$, $\mathbb{Q}$ and~$S$ can be constructed by a recursive procedure.
\end{proposition}

The notation $\sigmac(\mathbb{R}^d)$, $(\xi\cdot S)_\chi$ and $S^{\chi\ast}$ used in Proposition~\ref{Prop:20130727:pi-ag-dual} is defined by (\ref{eq:2012-10-03:sigmai}), (\ref{eq:20130726-X-chi-star}) and~(\ref{eq:20130726-X-stopped-at-chi}). The proof is provided in Appendix~\ref{Sect:Appendix:technical:seller}.

For any $\chi\in\mathcal{X}$ denote by $\bar{\mathcal{P}}^\mathrm{d}(\chi)$ the set of pairs $(\mathbb{Q},S)$ such that $\mathbb{Q}$ is a probability measure and $S$ is an $\mathbb{R}^d$-valued adapted process satisfying~(\ref{eq:chi-approx-u}).
Also define for $\cc=1,\ldots,d$
\[
  \Pbaruc(\chi):=\{(\mathbb{Q},S)\in\bar{\mathcal{P}}^\mathrm{d}(\chi)\st S^\cc_t=1\text{ for all }t=0,\ldots,T\}.
\]
The lack of arbitrage opportunities and Proposition~\ref{prop:mart-ito-Q} {ensure} that
\begin{align*}
 \emptyset &\neq \mathcal{P} \subseteq \bar{\mathcal{P}}^\mathrm{d}(\chi),&
 \emptyset &\neq \mathcal{P}_\cc \subseteq \Pbaruc(\chi)
\end{align*}
for all $\chi\in\mathcal{X}$.

\begin{remark}\upshape
The superscript $\mathrm{d}$ indicating deferred solvency distinguishes $\bar{\mathcal{P}}^\mathrm{d}(\chi)$ and $\Pbaruc(\chi)$ from the collections $\bar{\mathcal{P}}(\chi)$ and $\Pbarc(\chi)$ defined by \cite{Roux_Zastawniak2011} in a similar way as above, but with the weaker condition
\[
S_t\in\mathcal{K}^\ast_t\setminus\{0\}\ \ \text{and}\ \,\, \mathbb{E}_\mathbb{Q}(S^{\chi\ast}_{t+1}|\mathcal{F}_t)\in\mathcal{K}^\ast_t
\quad\text{for each }t=0,\ldots,T
\]
in place of~\eqref{eq:chi-approx-u}.
\end{remark}

The following result provides a representation of $\pi^\mathrm{ag}_\cc(\xi)$ dual to the representation~(\ref{eq:20130728:def-pi-ag}) in terms of superhedging strategies.

\begin{theorem}\label{Thm:ask-price-representation}
 The ask price in currency $j=1,\ldots,d$ of an American option~$\xi$ with gradual exercise can be written as
 \[
  \pi^\mathrm{ag}_\cc(\xi) = \max_{\chi\in\mathcal{X}}\max_{(\mathbb{Q},S)\in\Pbaruc(\chi)}\mathbb{E}_\mathbb{Q}((\xi\cdot S)_\chi).
 \]
 Moreover, we can algorithmically construct $\hat{\chi}\in\mathcal{X}$,  and $(\hat{\mathbb{Q}},\hat{S})\in\Pbaruc(\hat{\chi})$ such that
 \[
  \pi^\mathrm{ag}_\cc(\xi) = \mathbb{E}_{\hat{\mathbb{Q}}}((\xi\cdot \hat{S})_{\hat{\chi}}).
 \]
\end{theorem}

This theorem is proved in Appendix~\ref{Sect:Appendix:technical:seller}.

\begin{example}
\upshape We continue working in the setting of Example~\ref{Exl:new}. The
mixed stopping time $\hat{\chi}\in\mathcal{X}$ and a pair $(\mathbb{\hat{Q}%
},\hat{S})\in\bar{\mathcal{P}}^\mathrm{d}_{\!2}(\hat{\chi})$ such that
\[
\pi_{2}^{\mathrm{ag}}(\xi)=\mathbb{E}_{\hat{\mathbb{Q}}}((\xi\cdot\hat
{S})_{\hat{\chi}})=5
\]
are%
\[%
\begin{tabular}
[c]{|c|c|c|c|c|c|c|c|}\hline
$\rule[-0.2cm]{0pt}{0.6cm}$ & $\mathbb{\hat{Q}}$ & $\hat{S}_{0}$ & $\hat
{S}_{1}$ & $\hat{S}_{2}$ & $\hat{\chi}_{0}$ & $\hat{\chi}_{1}$ & $\hat{\chi
}_{2}$\\\hline
$\omega_{1}\rule[-0.2cm]{0pt}{0.6cm}$ & $1$ & $(5,1)$ & $(4,1)$ & $(8,1)$ &
$0$ & $\frac{3}{4}$ & $\frac{1}{4}$\\\cline{2-2}\cline{5-5}%
$\omega_{2}\rule[-0.2cm]{0pt}{0.6cm}$ & $0$ & $(5,1)$ & $(4,1)$ & $(4,1)$ &
$0$ & $\frac{3}{4}$ & $\frac{1}{4}$\\\cline{0-2}\cline{4-5}\cline{7-8}%
$\omega_{3}\rule[-0.2cm]{0pt}{0.6cm}$ & $0$ & $(5,1)$ & $(2,1)$ & $(3,1)$ &
$0$ & $0$ & $1$\\\cline{2-2}\cline{5-5}%
$\omega_{4}\rule[-0.2cm]{0pt}{0.6cm}$ & $0$ & $(5,1)$ & $(2,1)$ & $(1,1)$ &
$0$ & $0$ & $1$\\\hline
\end{tabular}
\]
\end{example}

\section{Pricing and superhedging for the buyer under gradual exercise}\label{Sect:buyer}

The buyer of an American option $\xi$ is entitled to receive the payoff according to a mixed stopping time $\chi\in\mathcal{X}$ of his choosing. In other words, the buyer receives~$\chi_t\xi_t$ at each time $t=0,\ldots,T$. The family $\Phi^\mathrm{bg}(\xi)$ of superhedging strategies for the buyer and the bid price (buyer's price) $\pi^\mathrm{bg}_\cc(\xi)$ under gradual exercise are defined in Section~\ref{Sect:grad-exe}. We turn to the task of computing the bid price and an optimal superhedging strategy for the buyer.

\subsection{Construction of buyer's price and superhedging strategy}

We start by computing the set if initial endowments that allow superhedging for the buyer.

\begin{construction}\label{const:2012-07-26.1}
 Construct adapted sequences $\mathcal{U}^\mathrm{bd}_t$, $\mathcal{V}^\mathrm{bd}_t$, $\mathcal{W}^\mathrm{bd}_t$, $\mathcal{Z}^\mathrm{bd}_t$ for $t=0,\ldots,T$ by
 \begin{align}
  \mathcal{U}^\mathrm{bd}_t&:=-\xi_t+\mathcal{Q}_t \text{ for all }t=0,\ldots,T,\label{eq:mathcal-Ut}\\
  \mathcal{Z}^\mathrm{bd}_T&:=
  \mathcal{V}^\mathrm{bd}_T:=
  \mathcal{W}^\mathrm{bd}_T:=
  \mathcal{U}^\mathrm{bd}_T,\nonumber\\
  \intertext{and by backward induction on all $t<T$}
  \mathcal{W}^\mathrm{bd}_t&:=\mathcal{Z}^\mathrm{bd}_{t+1}\cap\mathcal{L}_t,\label{eq:mathcal-Wt}\\
  \mathcal{V}^\mathrm{bd}_t&:=\mathcal{W}^\mathrm{bd}_t+\mathcal{Q}_t,\label{eq:mathcal-Vt}\\
  \mathcal{Z}^\mathrm{bd}_t&:=\conv\left\{\mathcal{U}^\mathrm{bd}_t,\mathcal{V}^\mathrm{bd}_t\right\}.\label{eq:mathcal-Zt}
 \end{align}
\end{construction}

For each $t$ the convex hull in \eqref{eq:mathcal-Zt} is taken on each atom of $\mathcal{F}_t$, i.e.\ for all $\mu\in\Omega_t$
\[
 \mathcal{Z}^{\mathrm{bd}\mu}_t = \conv\left\{\mathcal{U}^{\mathrm{bd}\mu}_t,\mathcal{V}^{\mathrm{bd}\mu}_t\right\}.
\]
The index $\mathrm{d}$ indicates that the deferred solvency cones $\mathcal{Q}_t$ are used in this construction.
The sets $\mathcal{U}^\mathrm{bd}_t$, $\mathcal{V}^\mathrm{bd}_t$, $\mathcal{W}^\mathrm{bd}_t$, $\mathcal{Z}^\mathrm{bd}_t$ are non-empty for each $t=0,\ldots,T$ because the portfolio $(z,\ldots,z)\in\mathbb{R}^d$ belongs to all of them when $z\in\mathbb{R}$ is large enough.

In contrast with Construction~4.6 of \cite{Roux_Zastawniak2011}, which was used the case of instant exercise at an ordinary stopping time, we have the convex hull of $\mathcal{U}^\mathrm{bd}_t,\mathcal{V}^\mathrm{bd}_t$ in~(\ref{eq:mathcal-Zt}) rather than the union of sets. This means that $\mathcal{U}^\mathrm{bd}_t$, $\mathcal{V}^\mathrm{bd}_t$, $\mathcal{W}^\mathrm{bd}_t$, $\mathcal{Z}^\mathrm{bd}_t$ are convex sets, unlike their counterparts in Construction~4.6 of \cite{Roux_Zastawniak2011}. This is important because, once it is established in the next proposition that the~$\mathcal{Z}^\mathrm{bd}_t$ are polyhedral, it becomes possible to implement techniques from convex analysis to compute them.

 \begin{proposition}\label{prop:2012-09-19:zt-closed}
  The set~$\mathcal{Z}^\mathrm{bd}_t$ in Construction~\ref{const:2012-07-26.1} is polyhedral with recession cone~$\mathcal{Q}_t$ for each $t=0,\ldots,T$.
 \end{proposition}

The proof of Proposition~\ref{prop:2012-09-19:zt-closed} can be found in Appendix~\ref{Sect:Appendix:technical:buyer}.

The next result shows that Construction~\ref{const:2012-07-26.1} produces the set of initial endowments that superhedges~$\xi$ for the buyer, which in turn makes it possible to compute the option bid price and also to construct a strategy that realises this price. This is similar to Theorem~\ref{prop:seller:Zau0=initial-endowments} for the seller.

\begin{theorem}\label{prop:2012-07-26:hedging-construct}
 The set of initial endowments that superhedge the buyer's position in the American option~$\xi$ with gradual exercise is equal to
\begin{equation}\label{eq:2012-07-26:constr-equivalence}
 \mathcal{Z}^\mathrm{bd}_0
  = \{x\in\mathbb{R}^d\st \exists (y,\chi)\in\Phi^\mathrm{bg}(\xi):x=y_0\},
\end{equation}
and the bid (buyer's) price of the option in currency $j=1,\ldots,d$ can be computed as
\[
\pi_{\cc}^{\mathrm{bg}}(\xi)=\max\left\{  -x\in\mathbb{R}\st xe^{\cc}%
\in\mathcal{Z}_{0}^{\mathrm{bd}}\right\}  .
\]
Moreover, a strategy $(y,\chi)\in\Phi^\mathrm{bg}(\xi)$ can be constructed such that
\[
\pi^\mathrm{bg}_\cc(\xi)e^\cc=-y_0.
\]
\end{theorem}

The proof of this theorem is also in Appendix~\ref{Sect:Appendix:technical:buyer}.

\begin{example}
\upshape Still within the setting of Example~\ref{Exl:new}, we apply the
constructions described in the current section to compute the sets
$\mathcal{Z}_{t}^{\mathrm{bd}}$ of superhedging portfolios for the buyer.
These are sets of portfolios $(x^{1},x^{2})\in\mathbb{R}^{2}$ satisfying the
inequalities%
\[%
\begin{tabular}
[c]{|c|c|c|c|}\hline
$\rule[-0.2cm]{0pt}{0.6cm}$ & $\mathcal{Z}_{0}^{\mathrm{bd}}$ & $\mathcal{Z}%
_{1}^{\mathrm{bd}}$ & $\mathcal{Z}_{2}^{\mathrm{bd}}$\\\hline
$\omega_{1}\rule[-0.6cm]{0cm}{1.45cm}$ & $5x^{1}+x^{2}\geq-3$ & $%
\begin{array}
[c]{c}%
8x^{1}+x^{2}\geq-8\\
6x^{1}+x^{2}\geq-4\\
4x^{1}+x^{2}\geq-4
\end{array}
$ & $%
\begin{array}
[c]{l}%
8x^{1}+x^{2}\geq-8\\
4x^{1}+x^{2}\geq0
\end{array}
$\\\cline{1-1}\cline{4-4}%
$\omega_{2}\rule[-0.6cm]{0cm}{1.45cm}$ & $5x^{1}+x^{2}\geq-3$ & $%
\begin{array}
[c]{c}%
8x^{1}+x^{2}\geq-8\\
6x^{1}+x^{2}\geq-4\\
4x^{1}+x^{2}\geq-4
\end{array}
$ & $4x^{1}+x^{2}\geq0$\\\cline{1-1}\cline{3-4}%
$\omega_{3}\rule[-0.6cm]{0cm}{1.45cm}$ & $5x^{1}+x^{2}\geq-3$ & $2x^{1}%
+x^{2}\geq0$ & $3x^{1}+x^{2}\geq0$\\\cline{1-1}\cline{4-4}%
$\omega_{4}\rule[-0.6cm]{0cm}{1.45cm}$ & $5x^{1}+x^{2}\geq-3$ & $2x^{1}%
+x^{2}\geq0$ & $x^{1}+x^{2}\geq0$\\\hline
\end{tabular}
\]
From $\mathcal{Z}_{0}^{\mathrm{bd}}$ we obtain the ask price
\[
\pi_{2}^{\mathrm{bg}}(\xi)=\max\{-x\in\mathbb{R}\st(0,x)\in\mathcal{Z}%
_{0}^{\mathrm{bd}}\}=3.
\]
We can also construct a superhedging strategy $(y,\chi)\in\Phi^{\mathrm{bg}%
}(\xi)$ such that
\[
(0,\pi_{2}^{\mathrm{bg}}(\xi))=-y_{0}=(0,3).
\]
It is the strategy~$(y,\chi)$ specified in Example~\ref{Exl:new}.
\end{example}

\subsection{Dual representation of buyer's price}

Since the $\mathcal{U}_{t}^{\mathrm{bd}}$, $\mathcal{V}_{t}^{\mathrm{bd}}$, $\mathcal{W}_{t}^{\mathrm{bd}}$, $\mathcal{Z}_{t}^{\mathrm{bd}}$ are convex,
it becomes possible to apply convex duality methods not just in the seller's
case but also in the buyer's case. (This was impossible to do in
\cite{Roux_Zastawniak2011} for American options with instant
exercise because of the lack of convexity in the buyer's case.)

In particular, in a similar way as in the proof of
Proposition~\ref{Prop:20130727:pi-ag-dual}, we can show that the bid price of
an American option with payoff~$\xi$ under gradual exercise can be expressed
as%
\[
\pi_{\cc}^{\mathrm{bg}}(\xi)=\max\left\{  Z_{0}^{\mathrm{bd}}(s)\st s\in
\sigma_{j}(\mathbb{R}^{d})\right\}
\]
in terms of the support function~$Z_{0}^{\mathrm{bd}}$ of $-\mathcal{Z}%
_{0}^{\mathrm{bd}}$.

However, we follow a different approach to obtain a
representation of the bid price $\pi_{\cc}^{\mathrm{bg}}(\xi)$ dual to the
representation~(\ref{eq:buyer-bid-price-gradual}) of $\pi_{\cc}^{\mathrm{bg}}(\xi)$
by means of superhedging strategies.
In Theorem~\ref{prop:2012-07-26:hedging-construct} a mixed
stopping time $\chi\in\mathcal{X}$ has already been constructed as part of a
superhedging strategy $(y,\chi)\in\Phi^{\mathrm{bg}}(\xi)$ such that
$\pi_{\cc}^{\mathrm{bg}}(\xi)=y_{0}$. As a result, the bid price given
by~(\ref{eq:buyer-bid-price-gradual}) can be written as%
\[
\pi_{\cc}^{\mathrm{bg}}(\xi)=\max\left\{  -x\in\mathbb{R}\st\exists y\in
\Phi:(y,\chi)\in\Phi^{\mathrm{bg}}(\xi),xe^{\cc}=y_{0}\right\}
\]
for this mixed stopping time~$\chi$. It turns out that the set on the
right-hand side can be expressed by means of the family $\Psi^\mathrm{a}(-\xi_\chi)$ of superhedging strategies for the seller of a
European option with expiry time $T$ and payoff $-\xi_\chi$ as described in Appendix~\ref{Sect:Eur-opt}, where $\xi_\chi$ is defined by~{\upshape(\ref{eq:20130726-X-stopped-at-chi})}.

\begin{proposition}\label{Prop:Am-Eur}
For any American option~$\xi$ and any mixed stopping time $\chi
\in\mathcal{X}$ we have%
\begin{multline*}
\{-x\in\mathbb{R}^{d}\st\exists y\in\Phi:(y,\chi)\in\Phi^{\mathrm{bg}}%
(\xi),x=y_{0}\}\\
=\{-x\in\mathbb{R}^{d}\st\exists z\in\Psi^{\mathrm{a}}(-\xi_\chi):x=z_{0}\}.
\end{multline*}
\end{proposition}

This proposition is proved in Appendix~\ref{Sect:Appendix:technical:buyer}.

We are now in a position to state a representation of the bid price dual to~(\ref{eq:buyer-bid-price-gradual}), and to prove it with the aid of Proposition~\ref{Prop:Am-Eur}.

\begin{theorem}\label{th:bu-buyer}
The buyer's (bid) price of an American option $\xi$ in currency $j=1,\ldots,d$ can be represented as
\begin{equation}\label{eq:pi-bu}
 \pi^\mathrm{bg}_\cc(\xi)
 = \max_{\chi\in\mathcal{X}}\min_{(\mathbb{Q},S)\in\Pbarc}\mathbb{E}_\mathbb{Q}\left((\xi\cdot S)_\chi\right),
\end{equation}
where $(\xi\cdot S)_\chi$ is defined by~{\upshape(\ref{eq:20130726-X-stopped-at-chi})}. Moreover, we can algorithmically construct $\hat{\chi}\in\mathcal{X}$ and $(\hat{\mathbb{Q}},\hat{S})\in\Pbarc$ such that
\[
 \pi^\mathrm{bg}_\cc(\xi)= \mathbb{E}_{\hat{\mathbb{Q}}}((\xi\cdot \hat{S})_{\hat{\chi}})
 =\min_{(\mathbb{Q},S)\in\Pbarc}\mathbb{E}_\mathbb{Q}\left((\xi\cdot S)_{\hat{\chi}}\right).
\]
\end{theorem}

The proof of Theorem~\ref{th:bu-buyer} is in Appendix~\ref{Sect:Appendix:technical:buyer}.

\begin{example}
\upshape We revisit Example~\ref{Exl:new} one more time to construct a mixed
stopping time $\hat{\chi}\in\mathcal{X}$ and a pair $(\mathbb{\hat{Q}},\hat
{S})\in\bar{\mathcal{P}}_{2}$ such that
\[
\pi_{2}^{\mathrm{bg}}(\xi)=\mathbb{E}_{\hat{\mathbb{Q}}}((\xi\cdot\hat
{S})_{\hat{\chi}})=\min_{(\mathbb{Q},S)\in\bar{\mathcal{P}}_{2}}%
\mathbb{E}_{\mathbb{Q}}((\xi\cdot S)_{\hat{\chi}})=3.
\]
They are%
\[%
\begin{tabular}
[c]{|c|c|c|c|c|c|c|c|}\hline
$\rule[-0.2cm]{0pt}{0.6cm}$ & $\mathbb{\hat{Q}}$ & $\hat{S}_{0}$ & $\hat
{S}_{1}$ & $\hat{S}_{2}$ & $\hat{\chi}_{0}$ & $\hat{\chi}_{1}$ & $\hat{\chi
}_{2}$\\\hline
$\omega_{1}\rule[-0.2cm]{0pt}{0.6cm}$ & $1$ & $(5,1)$ & $(5,1)$ & $(5,1)$ &
$0$ & $\frac{1}{2}$ & $\frac{1}{2}$\\\cline{2-2}\cline{5-5}%
$\omega_{2}\rule[-0.2cm]{0pt}{0.6cm}$ & $0$ & $(5,1)$ & $(5,1)$ & $(4,1)$ &
$0$ & $\frac{1}{2}$ & $\frac{1}{2}$\\\cline{0-2}\cline{4-5}\cline{7-8}%
$\omega_{3}\rule[-0.2cm]{0pt}{0.6cm}$ & $0$ & $(5,1)$ & $(2,1)$ & $(3,1)$ &
$0$ & $0$ & $1$\\\cline{2-2}\cline{5-5}%
$\omega_{4}\rule[-0.2cm]{0pt}{0.6cm}$ & $0$ & $(5,1)$ & $(2,1)$ & $(1,1)$ &
$0$ & $0$ & $1$\\\hline
\end{tabular}
\]
\end{example}

\section{Numerical example}\label{sec:num-example}

In this section we present a three-dimensional numerical example with a realistic flavour to illustrate Constructions~\ref{const:am:seller} and~\ref{const:2012-07-26.1}. The numerical procedures below were implemented in \emph{Maple} with the aid of the \emph{Convex} package \citep{Franz2009}.

Consider a model involving a domestic currency and two foreign currencies, with time horizon $\tau=1$ and with \textbf{$T=10$} time steps. The friction-free nominal exchange rates $E_t=(E^1_t,E^2_t)$ between the domestic currency and the two foreign currencies follow the two-asset {recombinant} \citet{Korn_Muller2009} model with Cholesky decomposition. That is, there are $(t+1)^2$ possibilities for the exchange rates at each time step $t=0,\ldots,T$, indexed by pairs $(j_1,j_2)$ with $1\le j_1,j_2\le t+1$, and each non-terminal node with exchange rates $E_t(j_1,j_2)$ has four successors, associated with exchange rates $E_{t+1}(j_1,j_2)$, $E_{t+1}(j_1+1,j_2)$, $E_{t+1}(j_1,j_2+1)$ and $E_{t+1}(j_1+1,j_2+1)$.  With $\Delta=\tfrac{\tau}{T}$ defined for convenience, the exchange rates are given by
\begin{align*}
 E^1_t(j_1,j_2) &= E^1_0e^{-\tfrac{1}{2}\sigma_1^2t\Delta +(2j_1-t-2)\sigma_1\sqrt{\Delta}},\\
 E^2_t(j_1,j_2) &= E^2_0e^{-\tfrac{1}{2}\sigma_2^2t\Delta +\left((2j_1-t-2)\rho + (2j_2-t-2)\sqrt{1-\rho^2}\right)\sigma_2\sqrt{\Delta}}
\end{align*}
for $t=0,\ldots,T$ and $j_1,j_2=1,\ldots,t+1$, where $E^1_0=40$ and $E^2_0=50$ are the initial exchange rates, $\sigma_1=15\%$ and $\sigma_2=10\%$ are the volatilities and $\rho=50\%$ is the correlation between the logarithmic growth of the exchange rates.

Assume that proportional transaction costs of $0.5\%$ are payable on all currency exchanges, except at time step $1$, when $10\%$ is payable, modelling a temporary loss of liquidity. In other words, the matrix of exchange rates between each pair among the three currencies at each time step $t$ is
\[
 \begin{pmatrix}\pi^{11}_t&\pi^{12}_t&\pi^{13}_t\\ \pi^{21}_t&\pi^{22}_t&\pi^{23}_t\\ \pi^{31}_t&\pi^{32}_t&\pi^{33}_t\end{pmatrix}
 =\begin{pmatrix}1&\frac{E^2_t}{E^1_t}(1+k_t)&\frac{1}{E^1_t}(1+k_t)\\ \frac{E^1_t}{E^2_t}(1+k_t)&1&\frac{1}{E^2_t}(1+k_t)\\ E^1_t(1+k_t)&E^2_t(1+k_t)&1\end{pmatrix},
\]
where
\[
 k_t =
 \begin{cases}
  0.1 &\text{if }t=1,\\
  0.005 &\text{otherwise}.
 \end{cases}
\]

Consider an American put option with physical delivery and strike~$90$ on a basket containing one unit of each of the foreign currencies. It offers the payoff
\[
 \xi_t=(-1,-1,90)\quad\text{for }t=0,\ldots,T.
\]
We allow for the possibility that the option may never be exercised by adding an extra time step $T+1$ to the model and setting the payoff to be $(0,0,0)$ at that time step. Constructions~\ref{const:am:seller} and~\ref{const:2012-07-26.1} give
\begin{multline*}
\begin{aligned}
\mathcal{Z}^\text{ad}_0 &= \conv\left\{
\left(\begin{array}{d}-0.749\\ -0.218\\ 47.587\end{array}\right),
\left(\begin{array}{d}-0.166\\ -0.727\\ 49.773\end{array}\right)
\right\} + \mathcal{Q}_0,\\
\mathcal{Z}^\text{bd}_0 &= \conv\left\{
\left(\begin{array}{d}0.631\\ 0.783\\ -65.310\end{array}\right),
\left(\begin{array}{d}0.453\\ 0.286\\ -33.404\end{array}\right),
\left(\begin{array}{d}0.419\\ 0.782\\ -56.798\end{array}\right),
\left(\begin{array}{d}0.249\\ 0.267\\ -24.342\end{array}\right),\right.
\end{aligned}\\
\left.
\left(\begin{array}{d}0.329\\ 0.513\\ -39.790\end{array}\right)
\right\} + \mathcal{Q}_0,
\end{multline*}
where $\mathcal{Q}_0$ is the convex cone generated by the vectors
\[
\left(\begin{array}{d}0.000 \\0.020 \\-1.000\end{array}\right),
\left(\begin{array}{d}-1.000 \\0.804 \\0.000\end{array}\right),
\left(\begin{array}{d}1.246 \\-1.000 \\0.000\end{array}\right),
\left(\begin{array}{d}-1.000 \\0.000 \\40.200\end{array}\right),
\left(\begin{array}{d}0.025 \\0.000 \\-1.000\end{array}\right),
\left(\begin{array}{d}0.000 \\-1.000 \\50.250\end{array}\right).
\]
The sets $\mathcal{Z}^\text{ad}_0$ and $\mathcal{Z}^\text{bd}_0$, which appear in Figure~\ref{fig:2}, yield the ask and bid prices
\begin{align*}
\pi_1^\text{ag}(\xi) &= 0.174, &\pi_2^\text{ag}(\xi) &= 0.140, & \pi_3^\text{ag}(\xi) &= 6.941,\\
\pi_1^\text{bg}(\xi) &= 0.022, &\pi_2^\text{bg}(\xi) &= 0.017, & \pi_3^\text{bg}(\xi) &= 0.879
\end{align*}
in each of the three currencies.
\begin{figure}
\begin{center}%
\begin{tabular}
[c]{cc}%
\includegraphics[width=0.5\figurewidth]{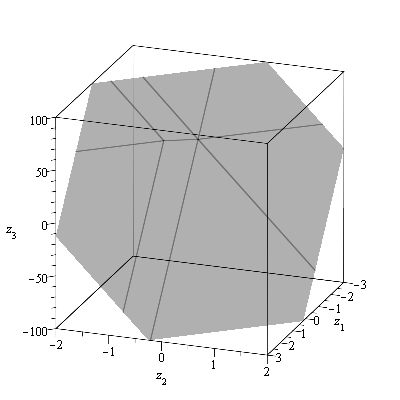} &
\includegraphics[width=0.5\figurewidth]{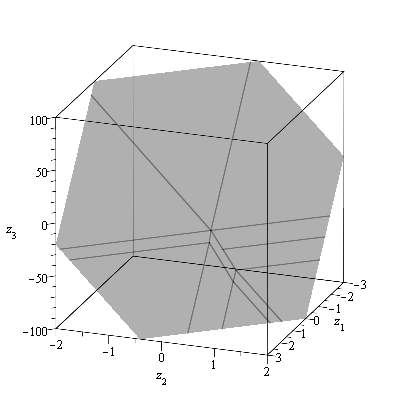} \\
$\mathcal{Z}^\text{ad}_0$ & $\mathcal{Z}^\text{bd}_0$
\end{tabular}
\end{center}
\caption{{Boundaries of the sets} of superhedging endowments for the seller and buyer}%
\label{fig:2}%
\end{figure}

\section{Conclusions and Outlook}\label{Sect:conclusions}

In this paper we have explored American options with gradual exercise within Kabanov's model
\citep{kabanov1999} of many assets under transaction costs, along with the
related notion of deferred solvency, which helps to deal with a temporary
loss of liquidity (large bid-ask spreads) in the market. We have demonstrated that gradual
exercise (at a mixed stopping time chosen by the buyer) can reduce the ask
(seller's) price and increase the bid (buyer's) price of the option compared
with the case when the option can only be exercised instantly (at an ordinary stopping time).

In this context we have constructed and implemented algorithms
to compute the ask and bid option
prices, the buyer's and seller's optimal hedging portfolios and strategies, and
their optimal mixed stopping times. We have studied dual representations for both the buyer and the seller of an American option with gradual
exercise. The results have been illustrated by numerical examples.

Compared to options with instant exercise, a novel feature is that pricing and
hedging an American option is a convex optimisation problem not just for the
seller but also for the buyer of the option, making it possible to use convex
duality in both cases. Ramifications to be explored further may
include an extension of Bouchard and Temam's representation of the strategies
hedging the seller's (short) position \citep{Bouchard_Temam2005} to the case
of hedging the buyer's (long) position in the option.

We also conjecture that it should be possible to adapt the constructions
presented here so that linear vector optimisation methods can be used to price and
hedge both the seller's and buyer's positions in an American option with
gradual exercise, along similar lines as was done by
\cite{Loehne_Rudloff2011} for European options under transaction costs.

\appendix
\section{Appendix}

\subsection{Some notation and facts from convex analysis\label{Subsect:conv-anal}}

For any non-empty convex cone $A\subseteq\mathbb{R}^d$, denote by $A^\ast$ the \emph{polar} of $-A$, i.e.
\begin{equation}\label{eq:2012-09-20:Aast}
 A^\ast:=\{y\in\mathbb{R}^d \st y\cdot x\ge0 \text{ for all }x\in A\}.
\end{equation}
For any set $A\subseteq\mathbb{R}^d$ define the \emph{cone generated by $A$} as
\[
 \cone A := \{\lambda x \st \lambda\ge0,x\in A\}.
\]
The \emph{recession cone} of a non-empty convex set $A\subseteq\mathbb{R}^d$ is defined as
\[
 0^+A:=\{y\in\mathbb{R}^d \st A +\lambda y\subseteq A \text{ for all } \lambda\ge0\}.
\]
It is a convex cone containing the origin \citep[Theorem 8.1]{rockafellar1996}. If $A$ is a polyhedral cone, then $0^+A=A$ \citep[Corollary 8.3.2]{rockafellar1996}.

The \emph{convex hull} of sets $A_1,\ldots,A_n$ in $\mathbb{R}^d$ is the smallest convex set in $\mathbb{R}^d$ that contains $A_1,\ldots,A_n$, and is denoted by $\conv\{A_1,\ldots,A_n\}$. The \emph{convex hull} of convex functions $f_1,\ldots,f_n:\mathbb{R}\rightarrow\mathbb{R}\cup\{\infty\}$ is the function $f:=\conv \{f_1,\ldots,f_n\}$ defined by
\[
 f(x) := \inf\left\{\sum_{i=1}^n\lambda_i f_i(x_i) \st
 0\le\lambda_1,\ldots,\lambda_n\le1,
 \sum_{i=1}^n\lambda_i=1,
 \sum_{i=1}^n\lambda_i x_i=x\right\}.
\]
The \emph{effective domain} of a convex function $f:\mathbb{R}^d\rightarrow\mathbb{R}\cup\{\infty\}$ is defined as
\[
 \dom f:= \{y\in\mathbb{R}^d \st f(y)<\infty\}.
\]
The \emph{support function} $\delta^\ast_A$ of a convex set $A\subseteq\mathbb{R}^d$ is defined as
\[
 \delta^\ast_A(x):=\sup\{x\cdot y \st y\in A\}.
\]

\subsection{Compactly generated cones\label{Subsect:comp-gen-cones}}

For any set $A\subseteq\mathbb{R}^d$ and $\cc=1,\ldots,d$ define
\begin{equation}\label{eq:2012-10-03:sigmai}
 \sigmac(A) := \{x=(x^1,\ldots,x^d)\in A \st x^\cc=1\}.
\end{equation}
We say that a cone $A\subset\mathbb{R}^d$ is \emph{compactly $\cc$-generated} if $\sigmac(A)$ is compact, non-empty and $A$ is generated by $\sigmac(A)$.

\begin{lemma} \label{lem:2012-10-03:sigmai-intersect}
 If two cones $A\subset\mathbb{R}^d$ and $B\subset\mathbb{R}^d$ are compactly $\cc$-generated and \mbox{$A\cap B\setminus \{0\}\neq\emptyset$}, then $A\cap B$ is compactly $\cc$-generated and
 \begin{equation}\sigmac(A\cap B) = \sigmac(A)\cap\sigmac(B).\label{eq:2012-10-03:sigmai-intersect} \end{equation}
\end{lemma}

\begin{proof}
 Equality \eqref{eq:2012-10-03:sigmai-intersect} follows directly from \eqref{eq:2012-10-03:sigmai}. A vector $x=(x^1,\ldots,x^d)$ is an element of $\sigmac(A\cap B)$ if and only if $x\in A\cap B$ and $x^\cc=1$, if and only if $x\in A$ and $x\in B$ and $x^\cc=1$, if and only if $x\in\sigmac(A)$ and $x\in\sigmac(B)$, if and only if $x\in\sigmac(A)\cap\sigmac(B)$.

 The set $\sigmac(A\cap B)$ is compact since it is the intersection of two compact sets $\sigmac(A)$ and $\sigmac(B)$. It remains to show that $\sigmac(A\cap B)$ is non-empty and generates $A\cap B$. To this end, fix any $x\in A\cap B\setminus\{0\}$. As $A$ and $B$ are generated, respectively, by $\sigmac(A)$ and $\sigmac(B)$, there exist $\lambda_A\ge0$, $\lambda_B\ge0$, $x_A\in\sigmac(A)$ and $x_B\in\sigmac(B)$ such that
 \[
  \lambda_Ax_A = x = \lambda_Bx_B.
 \]
 As $x\neq0$, we must have $\lambda_A>0$ and $\lambda_B>0$. Moreover, since $x_{\!A}^\cc=x_{\!B}^\cc=1$, we have
 \[
  \lambda_A = \lambda_Ax_{\!A}^\cc = x^\cc = \lambda_Bx_{\!B}^\cc = \lambda_B,
 \]
 which in turn implies $x_A=x_B\in\sigmac(A)\cap\sigmac(B)=\sigmac(A\cap B)$, completing the proof.
 \end{proof}

In this paper we also make use of the following result by \citet[Lemma~A.1]{Roux_Zastawniak2011}.

\begin{lemma}\label{lem:Roux_Zastawniak2011}
 Fix any $\cc=1,\ldots,d$, and suppose that $A_1,\ldots,A_n$ are non-empty closed convex sets in $\mathbb{R}^d$ such that $A:=\bigcap_{i=1}^nA_i\neq\emptyset$ and $\dom\delta^\ast_{A_i}$ is compactly $\cc$-generated for all $k$. Then
 \[
  \delta^\ast_A = \conv\{\delta^\ast_{A_1},\ldots,\delta^\ast_{A_n}\},
 \]
 the cone $\dom\delta^\ast_A$ is compactly $\cc$-generated and
 \[
  \dom\delta^\ast_A=\conv\{\dom\delta^\ast_{A_1}, \ldots, \dom\delta^\ast_{A_n}\},
 \]
 and for each $x\in\sigmac(\dom\delta^\ast_A)$ there exist $\alpha_1,\ldots,\alpha_n\ge0$ and $x_1,\ldots,x_n$ with $x_i\in\sigmac(\dom\delta^\ast_{A_i})$ for all $i=1,\ldots,n$ such that
 \begin{align*}
  \delta^\ast_A(x) &= \sum_{i=1}^n\alpha_i\delta^\ast_{A_i}(x_i), &
  x &= \sum_{i=1}^n\alpha_i x_i, &
  1 &= \sum_{i=1}^n\alpha_i.
 \end{align*}
\end{lemma}

\subsection{Mixed stopping times}\label{Sect:mixed-stop-times}

A \emph{mixed} (or \emph{randomised}) \emph{stopping time} is a non-negative adapted process $\chi=(\chi_t)_{t=0}^T$ with values in $[0,1]$ such that
\[
\sum_{t=0}^T\chi_t=1.
\]
The family of mixed stopping times will be denoted by~$\mathcal{X}$.

For any $\chi\in\mathcal{X}$ we put
\begin{align}\label{eq:2013-07-13-chi-star}
\chi^\ast_t&:=\sum_{s=t}^T\chi_s\text{ for }t=0,\ldots,T,
&
\chi^\ast_{T+1}&:=0.
\end{align}%
Moreover, for any adapted process $X$ and for any $\chi\in\mathcal{X}$ we put
\begin{align}\label{eq:20130726-X-chi-star}
X^{\chi\ast}_t&:=\sum_{s=t}^T\chi_s X_s\text{ for }t=0,\ldots,T,
&
X^{\chi\ast}_{T+1}&:=0.
\end{align}
We also define \emph{$X$ evaluated at $\chi$} by
\begin{align}\label{eq:20130726-X-stopped-at-chi}
X_\chi:=\sum_{s=0}^T\chi_s X_s.
\end{align}
With each ordinary stopping time $\tau\in\mathcal{T}$ we associate the mixed stopping time $\chi^\tau\in\mathcal{X}$ defined as
\begin{equation}\label{eq:20130812-chi-tau}
\chi_{t}^{\tau}:=\left\{
\begin{array}
[c]{cc}%
1 & \text{on }\{\tau=t\}\\
0 & \text{on }\{\tau\neq t\}
\end{array}
\right.  \quad\text{for each }t=0,\ldots,T.
\end{equation}

\subsection{Proofs and technical results}

\subsubsection{Deferred solvency}\label{Sect:Appendix:defer-solv}

\begin{proof}[\textbf{\upshape Proof of Proposition~\ref{prop:constr:ultimate-solvent}.}]
Equality (\ref{eq:QT=KT}) is obvious. By the definition of the deferred solvency cones, for any $t=0,\ldots,T-1$ the following conditions are equivalent: $z\in\mathcal{Q}_{t+1}\cap\mathcal{L}_t+\mathcal{K}_t$ if and only if there is a $y_{t+1}\in\mathcal{L}_t$ such that
\begin{align*}
y_{t+1}&\in\mathcal{Q}_{t+1},& z-y_{t+1}&\in\mathcal{K}_t,
\end{align*}
if and only if there is a sequence $y_{t+1},\ldots,y_{T+1}$ such that $y_s\in\mathcal{L}_{s-1}$ for each $s=t+1,\ldots,T$ and
\begin{align*}
z-y_{t+1}&\in\mathcal{K}_t,&
y_s-y_{s+1}&\in\mathcal{K}_s\mathrm{~for~all~}s=t+1\ldots,T,&
y_{T+1}&=0,
\end{align*}
if and only if $z\in\mathcal{Q}_t$. This proves~(\ref{eq:Qt-recursive}).
\end{proof}

\begin{proof}[\textbf{\upshape Proof of Proposition~\ref{prop:mart-ito-Q}.}]
In view of \eqref{eq:Qt-recursive}, we have $\mathcal{K}_t\subseteq\mathcal{Q}_t$, so $\mathcal{Q}^*_t\subseteq\mathcal{K}^*_t$, and (\ref{eq:2012-10-03:Qastt-not-empty}) implies~\eqref{eq:th:2012-10-03:ftap}.

Conversely, suppose that $S$ is an $\mathbb{R}^d$-valued $\mathbb{Q}$-martingale that satisfies~\eqref{eq:th:2012-10-03:ftap}. To show that it satisfies (\ref{eq:2012-10-03:Qastt-not-empty}) we proceed by backward induction. By \eqref{eq:QT=KT}, we have $S_T\in\mathcal{K}_T^\ast\setminus\{0\}=\mathcal{Q}_T^\ast\setminus\{0\}$. For any $t=0,\ldots,T-1$ suppose that $S_{t+1}\in\mathcal{Q}_{t+1}^\ast\setminus\{0\}$. As $S$ is a $\mathbb{Q}$-martingale, we have for all $\mu\in\Omega_t$ that
 \[
  S^\mu_t=\mathbb{E}_\mathbb{Q}(S_{t+1}\,|\,\mu)
  \in \conv\{\mathcal{Q}_{t+1}^{\nu\ast}\st \nu\in\successors\mu\}.
 \]
For every $t=0,\ldots,T-1$ and $\mu\in\Omega_t$, observe from \eqref{eq:2013-07-12-Q-mu-t} that
\[
 \mathcal{Q}^{\mu\ast}_t = \left[\bigcap_{\nu\in\successors\mu}\mathcal{Q}^\nu_{t+1}+\mathcal{K}_t^\mu\right]^\ast.
\]
Successive application of Corollaries 16.4.2 and 16.5.2 in \citet{rockafellar1996} then gives
\begin{equation} \label{eq:2012-10-03:Qast-ito-domWKt}
 \mathcal{Q}^{\mu\ast}_t = \left[\bigcap_{\nu\in\successors\mu}\mathcal{Q}^\nu_{t+1}\right]^\ast\cap\mathcal{K}^{\mu\ast}_t = \conv\{\mathcal{Q}^{\nu\ast}_{t+1}\st \nu\in\successors\mu\}\cap\mathcal{K}^{\mu\ast}_t.
\end{equation}
Since $S^\mu_t\in\mathcal{K}^{\mu\ast}_t\setminus\{0\}$ by \eqref{eq:th:2012-10-03:ftap}, it follows that $S^\mu_t\in\mathcal{Q}^{\mu\ast}_t\setminus\{0\}$, which concludes the inductive step.
\end{proof}

\subsubsection{Seller's case}\label{Sect:Appendix:technical:seller}

\begin{proof}[\textbf{\upshape Proof of Proposition~\ref{prop:am:seller:immediate-ultimate}.}]
  We show first that for any $y\in\Phi^\mathrm{ad}(\xi)$ there exists $Y\in\Phi^\mathrm{ag}(\xi)$ such that $Y_0=y_0$.
  If $y\in\Phi^\mathrm{ad}(\xi)$, then for each $t=0,\ldots,T-1$ we have $y_t-y_{t+1}\in\mathcal{Q}_t$, i.e.\ there exists a liquidation strategy $z^t_{t+1},\ldots,z^t_{T+1}$ starting from $y_t-y_{t+1}$ at time~$t$.
  We also put $z^T_{T+1}:=0$ for notational convenience.
  Moreover, for each $t=0,\ldots,T$ we have $y_t-\xi_t\in\mathcal{Q}_t$, i.e.\ there exists a liquidation strategy $x^t_{t+1},\ldots,x^t_{T+1}$ starting from $y_t-\xi_t$ at time~$t$. For each $\chi\in\mathcal{X}$ define
\begin{align*}
 Y^\chi_0 &:=y_0,\\
 Y^\chi_t &:= \chi^\ast_ty_t
 + \sum_{s=0}^{t-1} \chi^\ast_{s+1} z^s_t
 + \sum_{s=0}^{t-1} \chi_sx^s_t \text{ for } t=1,\ldots,T+1,
\end{align*}
where $\chi^\ast$ is defined by~(\ref{eq:2013-07-13-chi-star}).
The process $Y^\chi$ belongs to~$\Phi$ and satisfies the non-anticipation condition~(\ref{eq:non-anticipate}). Moreover, for each $t=0,\ldots,T$
\begin{align*}
 Y^\chi_t-\chi_t\xi_t-Y^\chi_{t+1}
&= \chi^\ast_ty_t + \sum_{s=0}^{t-1} \chi^\ast_{s+1} z^s_t
+ \sum_{s=0}^{t-1} \chi_sx^s_t - \chi_t\xi_t\\
&\quad\quad\quad\quad - \chi^\ast_{t+1}y_{t+1}
- \sum_{s=0}^{t} \chi^\ast_{s+1} z^s_{t+1}
 - \sum_{s=0}^t \chi_sx^s_{t+1} \\
&= \chi^\ast_{t+1}(y_t-y_{t+1}-z^t_{t+1}) +\chi_t(y_t-\xi_t-x^t_{t+1})\\
&\quad\quad\quad\quad + \sum_{s=0}^{t-1} \chi^\ast_{s+1}(z^s_t - z^s_{t+1})
+ \sum_{s=0}^{t-1} \chi_s(x^s_t - x^s_{t+1}) \\
&\in \chi^\ast_{t+1}\mathcal{K}_t +\chi_t\mathcal{K}_t
+ \sum_{s=0}^{t-1} \chi^\ast_{s+1}\mathcal{K}_t
+ \sum_{s=0}^{t-1} \chi_s\mathcal{K}_t
 \subseteq \mathcal{K}_t
\end{align*}
because $\chi^\ast_t=\chi_t+\chi^\ast_{t+1}$ and $\mathcal{K}_t$ is a convex cone. Hence $Y$  satisfies (\ref{eq:seller-self-fin-superhedge}) in addition to~(\ref{eq:non-anticipate}), and so $Y\in\Phi^\mathrm{ag}(\xi)$.

Conversely, fix any $Y\in\Phi^\mathrm{ag}(\xi)$ and put $y:=Y^{\chi^T}\in\Phi$, where $\chi^T$ is defined by~(\ref{eq:20130812-chi-tau}). Then for all {$s=0,\ldots,T-1$} we have {$\chi^T_s=0$ and}
\[
 y_{s} - y_{{s}+1} = Y_{s}^{\chi^T} - Y^{\chi^T}_{{s}+1} = Y_{s}^{\chi^T} - \chi^T_{s}\xi_{s} - Y^{\chi^T}_{{s}+1} \in \mathcal{K}_{s}\subseteq\mathcal{Q}_{s}.
\]
Fix any $t=0,\ldots,T$. Then $\chi_{s}^{T}=\chi_{s}^{t}=0$ for each
$s=0,\ldots,t-1$, and the non-anticipation property~(\ref{eq:non-anticipate}) of $Y$ gives
$y_{t}=Y_{t}^{\chi^{T}}=Y_{t}^{\chi^{t}}$. Since $\chi_{t}^{t}=1$,
it means that
\begin{equation}\label{eq:2012-12-29_2}
y_{t}-\xi_{t}-Y_{t+1}^{\chi^{t}}
=Y_{t}^{\chi^{t}}-\chi_{t}^{t}\xi_{t}-Y_{t+1}^{\chi^{t}}
\in\mathcal{K}_{t}\subseteq\mathcal{Q}_{t}.
\end{equation}
Moreover, for each $s=t+1,\ldots,T$ we have $\chi_{s}^{t}=0$, and
so
\begin{equation}\label{eq:2012-12-29_1}
Y_{s}^{\chi^{t}}-Y_{s+1}^{\chi^{t}}
=Y_{s}^{\chi^{t}}-\chi_{s}^{t}\xi_{s}-Y_{s+1}^{\chi^{t}}
\in\mathcal{K}_{s}\subseteq\mathcal{Q}_{s}.
\end{equation}
We verify by backward induction that $Y_{s+1}^{\chi^{t}}\in\mathcal{Q}_{s}$
for each $s=t,\ldots,T$. Clearly, $Y_{T+1}^{\chi^{t}}=0\in\mathcal{Q}_{T}$.
Now suppose that $Y_{s+1}^{\chi^{t}}\in\mathcal{Q}_{s}$ for some
$s=t+1,\ldots,T$. From (\ref{eq:2012-12-29_1}) we can see that $Y_{s}^{\chi^{t}}=(Y_{s}^{\chi^{t}}-Y_{s+1}^{\chi^{t}})+Y_{s+1}^{\chi^{t}}\in\mathcal{Q}_{s}+\mathcal{Q}_{s}=\mathcal{Q}_{s}$.
Because $Y^{\chi^{t}}$ is predictable, we have $Y_{s}^{\chi^{t}}\in\mathcal{Q}_{s}\cap\mathcal{L}_{s-1}\subseteq\mathcal{Q}_{s-1}$
by~(\ref{eq:Qt-recursive}), completing the backward induction argument.
In particular, this means that $Y_{t+1}^{\chi^{t}}\in\mathcal{Q}_{t}$.
Together with (\ref{eq:2012-12-29_2}) it gives $y_{t}-\xi_{t}\in\mathcal{Q}_{t}$
for any $t=0,\ldots,T$. As a result, we have constructed $y\in\Phi^\mathrm{ad}(\xi)$ such that $y_0=Y_0$.
\end{proof}

\begin{proof}[\textbf{\upshape Proof of Theorem~\ref{prop:seller:Zau0=initial-endowments}.}]
 Suppose that $x\in\mathcal{Z}^\mathrm{ad}_0$. We construct a sequence $y=(y_t)_{t=0}^{T+1}$ of random variables by induction. First take $y_0=x$. Now suppose that we have already constructed $y_t\in\mathcal{Z}^\mathrm{ad}_t$ such that $y_t\in\mathcal{L}_{(t-1)\vee 0}$ for some $t=0,\ldots,T-1$. From \eqref{eq:seller:mathcal-Zt} we obtain $y_t\in\mathcal{U}^\mathrm{ad}_t$, whence
 \[
 y_t-\xi_t\in\mathcal{Q}_t
 \]
 by \eqref{eq:seller:mathcal-Ut}. We also obtain $y_t\in\mathcal{V}^\mathrm{ad}_t$, and by~(\ref{eq:seller:mathcal-Vt}) there exists a random variable $y_{t+1}\in\mathcal{W}^\mathrm{ad}_t$ such that $y_t-y_{t+1}\in\mathcal{Q}_t$. From \eqref{eq:seller:mathcal-Wt} we have $y_{t+1}\in\mathcal{Z}^\mathrm{ad}_{t+1}\cap\mathcal{L}_t$, which concludes the inductive step. Finally, we put $y_{T+1}:=0$. It follows that $y\in\Phi^\mathrm{ad}(\xi)$ with $y_0=x$. By Proposition~\ref{prop:am:seller:immediate-ultimate}, a strategy $Y\in\Phi^\mathrm{ag}(\xi)$ can be constructed such that $Y_0=y_0=x$.

 Suppose now that $Y\in\Phi^\mathrm{ag}(\xi)$. By Proposition~\ref{prop:am:seller:immediate-ultimate}, there is a $y\in\Phi^\mathrm{ad}(\xi)$ such that $Y_0=y_0$. Clearly,
 \begin{equation}\label{eq:seller:hedging:Ut}
 y_t\in\xi_t+\mathcal{Q}_t=\mathcal{U}^\mathrm{ad}_t
 \end{equation}
 for all $t=0,\ldots,T$, and in particular $y_T\in\mathcal{Z}^\mathrm{ad}_T=\mathcal{U}^\mathrm{ad}_T$. We now show by backward induction that $y_t\in\mathcal{Z}^\mathrm{ad}_t$ for all~$t=0,\ldots,T$. Suppose that $y_{t+1}\in\mathcal{Z}^\mathrm{ad}_{t+1}$ for some $t=0,\ldots,T-1$. Since $y_{t+1}\in\mathcal{L}_t$, this means by \eqref{eq:seller:mathcal-Wt} that $y_{t+1}\in\mathcal{W}^\mathrm{ad}_t$. The condition $y_t-y_{t+1}\in\mathcal{Q}_t$ implies that $y_t\in\mathcal{W}^\mathrm{ad}_t+\mathcal{Q}_t=\mathcal{V}^\mathrm{ad}_t$. Property~\eqref{eq:seller:hedging:Ut} then gives $y_t\in\mathcal{Z}^\mathrm{ad}_t$ by \eqref{eq:seller:mathcal-Zt}, which completes the inductive step. We conclude that $Y_0=y_0\in\mathcal{Z}^\mathrm{ad}_0$.

We have proved~(\ref{eq:20130727-Zad0}). It follows that
\begin{align*}
 \pi^\mathrm{ag}_\cc(\xi)
 &=\inf\{x\in\mathbb{R}\st \exists Y\in\Phi^\mathrm{ag}(\xi): xe^\cc=Y_0\}\\
 &=\inf\{x\in\mathbb{R}\st xe^\cc\in\mathcal{Z}^\mathrm{ad}_0\}.
\end{align*}
We know that $\mathcal{Z}^\mathrm{ad}_0$ is polyhedral, hence closed, so $\{x\in\mathbb{R}\st xe^\cc\in\mathcal{Z}^\mathrm{ad}_0\}$ is also a closed set. It is non-empty and bounded below because $xe^\cc\in\mathcal{Z}^\mathrm{ad}_0$ for any $x\in\mathbb{R}$ large enough, and $xe^\cc\notin\mathcal{Z}^\mathrm{ad}_0$ for any $x\in\mathbb{R}$ small enough. As a result, the infimum is attained. It means, in particular, that $\pi^\mathrm{ag}_\cc(\xi)e^\cc\in\mathcal{Z}^\mathrm{ad}_0$, so we know that a strategy $Y\in\Phi^\mathrm{ag}(\xi)$ can be constructed such that $\pi^\mathrm{ag}_\cc(\xi)e^\cc=Y_0$.
 \end{proof}

The following result is similar to Lemma~5.5 in \citep{Roux_Zastawniak2011}, but with the solvency cones~$\mathcal{K}_t$ replaced by the deferred solvency cones~$\mathcal{Q}_t$.

\begin{proposition}\label{prop:seller:dual} {\ }
\begin{enumerate}[$(1)$]
\item \label{item:prop:seller:dual:0} For each $t=0,\ldots,T$ the set $\mathcal{Q}^\ast_t$ is compactly $\cc$-generated.
\item \label{item:prop:seller:dual:1} For all $t=0,\ldots,T$ and $y\in\mathcal{L}_t$ we have
 \begin{align}
   U^\mathrm{ad}_t(y) &=
   \begin{cases}
    -y\cdot\xi_t & \text{on } \{y\in\mathcal{Q}^\ast_t\},\\
    \infty & \text{on } \{y\notin\mathcal{Q}^\ast_t\},
   \end{cases}\label{eq:seller:Uaut}\\
   V^\mathrm{ad}_t(y) &=
   \begin{cases}
    W^\mathrm{ad}_t(y) & \text{on } \{y\in\mathcal{Q}^\ast_t\},\\
    \infty & \text{on } \{y\notin\mathcal{Q}^\ast_t\}.
   \end{cases}\label{eq:seller:Vaut}
   \end{align}

   \item \label{item:prop:seller:dual:2} We have $\dom Z^\mathrm{ad}_t=\mathcal{Q}^\ast_t$ for all $t=0,\ldots,T$. For all $t=0,\ldots,T-1$ and $\mu\in\Omega_t$ we have
   \begin{equation}
      Z^\mathrm{ad\mu}_t = \conv\{U^\mathrm{ad\mu}_t,V^\mathrm{ad\mu}_t\},\label{eq:seller:Zaut}
    \end{equation}
    and for each $y\in\sigmac(\mathcal{Q}^{\mu\ast}_t)$ there exist $\lambda\in[0,1]$, $x\in\sigmac(\dom V^\mathrm{ad\mu}_t)$ and ${s}\in\sigmac(\mathcal{Q}^{\mu\ast}_t)$ such that
    \begin{align*}
      Z^\mathrm{ad\mu}_t(y) &= \lambda U^\mathrm{ad\mu}_t({s}) + (1-\lambda)V^\mathrm{ad\mu}_t(x),&
      y &= \lambda {s} + (1-\lambda)x.
    \end{align*}

    \item \label{item:prop:seller:dual:3}For every $t=0,\ldots,T-1$ and $\mu\in\Omega_t$ we have
   \[
      W^\mathrm{ad\mu}_t = \conv\{Z^\mathrm{ad\nu}_{t+1}
      \st \nu\in\successors\mu\},
    \]
    and for each $x\in\sigmac(\dom W^\mathrm{ad\mu}_t)$ there exist $q^\nu\in[0,1]$ and $y^\nu\in\sigmac(\mathcal{Q}^{\nu\ast}_{t+1})$ for all $\nu\in\successors\mu$ such that
    \begin{align*}
      W^\mathrm{ad\mu}_t(x) &= \sum_{\nu\in\successors\mu}q^\nu Z^\mathrm{ad\nu}_{t+1}(y^\nu),& x &= \sum_{\nu\in\successors\mu}q^\nu y^\nu, & 1 &= \sum_{\nu\in\successors\mu}q^\nu.
    \end{align*}
\end{enumerate}
\end{proposition}

\begin{proof}
We first consider claim~\eqref{item:prop:seller:dual:1}.
As $\mathcal{Q}_t$ is a cone,
\[
 \delta^\ast_{-\mathcal{Q}_t}(x)
  = \begin{cases}
     0 & \text{if } x\in\mathcal{Q}^\ast_t,\\
     \infty & \text{otherwise}.
    \end{cases}
\]
 Note in particular that $\dom\delta^\ast_{-\mathcal{Q}_t}=\mathcal{Q}^\ast_t$.
For all $t$ and $y\in\mathcal{L}_t$ we have \citep[p.~113]{rockafellar1996}
   \[
    U^\mathrm{ad}_t(y) = \delta^\ast_{-\xi_t-\mathcal{Q}_t}(y) = \delta^\ast_{\{-\xi_t\}}(y) + \delta^\ast_{-\mathcal{Q}_t}(y),
   \]
    which leads to~\eqref{eq:seller:Uaut}. Equation \eqref{eq:seller:Vaut} follows similarly from
   \[
    V^\mathrm{ad}_t = \delta^\ast_{-\mathcal{W}^\mathrm{ad}_t-\mathcal{Q}_t} = \delta^\ast_{-\mathcal{W}^\mathrm{ad}_t} + \delta^\ast_{-\mathcal{Q}_t} = W^\mathrm{ad}_t +\delta^\ast_{-\mathcal{Q}_t}. \]

    Claims \eqref{item:prop:seller:dual:0}, \eqref{item:prop:seller:dual:2} and \eqref{item:prop:seller:dual:3} can be obtained by backward induction. We clearly have
    \[
     \dom Z^\mathrm{ad}_T = \dom U^\mathrm{ad}_T = \dom \delta^\ast_{-\mathcal{Q}_T} = \mathcal{Q}^\ast_T,
    \]
    and this set is compactly $\cc$-generated because $\mathcal{Q}^\ast_T=\mathcal{K}^\ast_T$.

    Now fix any $t=0,\ldots,T-1$ and $\mu\in\Omega_t$, and suppose that $\dom Z^\mathrm{ad}_{t+1}=\mathcal{Q}^\ast_{t+1}$ and that this set is compactly $\cc$-generated. Since
    \[
     \bigcap_{\nu\in\successors\mu}\mathcal{Z}^{\mathrm{ad}\nu}_{t+1} = \mathcal{W}^{\mathrm{ad}\mu}_t\neq\emptyset,
    \]
    Lemma~\ref{lem:Roux_Zastawniak2011} can be applied to the sets $-\mathcal{Z}^{\mathrm{ad}\nu}_{t+1}$ for all $\nu\in\successors\mu$. This justifies claim \eqref{item:prop:seller:dual:3} for this~$t$ and also that
    \[
     \dom W^{\mathrm{ad}\mu}_t = \conv\left\{\dom Z^{\mathrm{ad}\nu}_{t+1}
     \st \nu\in\successors\mu\right\} = \conv\left\{\mathcal{Q}^{\nu\ast}_{t+1}
     \st \nu\in\successors\mu\right\}
    \]
    is compactly $\cc$-generated.

By Theorem~\ref{th:2012-10-03:ftap} and Proposition~\ref{prop:mart-ito-Q}, the lack of arbitrage opportunities implies that there is a pair $(\mathbb{Q},S)\in\mathcal{P}$ such that $S_{t}^{\mu}\in\mathcal{Q}_{t}^{\mu*}\setminus\{0\}\subseteq\mathcal{K}_{t}^{\mu*}\setminus\{0\}$ and
$S_{t+1}^{\nu}\in\mathcal{Q}_{t+1}^{\nu*}$
for each $\nu\in\successors\mu$. Since $S$ is a martingale under $\mathbb{Q}$, it follows that \[
S_{t}^{\mu}=\mathbb{E}_\mathbb{Q}(S_{t+1}\,|\,\mu)
\in\conv\{\mathcal{Q}_{t+1}^{\nu*}\st \nu\in\successors\mu\},
\]
and so
$S_{t}^{\mu}\in\conv\{\mathcal{Q}_{t+1}^{\nu*}\st \nu\in\successors\mu\}
\cap\mathcal{K}_{t}^{\mu*}\setminus\{0\}\ne\emptyset$.
As $\conv\left\{\mathcal{Q}^{\nu\ast}_{t+1}\st \nu\in\successors\mu\right\}$ and $\mathcal{K}^{\mu\ast}_t$ are compactly $\cc$-generated, it follows from Lemma~\ref{lem:2012-10-03:sigmai-intersect} and \eqref{eq:2012-10-03:Qast-ito-domWKt} that $\mathcal{Q}^{\mu\ast}_t=\conv\left\{\mathcal{Q}^{\nu\ast}_{t+1}
\st \nu\in\successors\mu\right\}\cap \mathcal{K}^{\mu\ast}_t$ is compactly $\cc$-generated, which justifies claim \eqref{item:prop:seller:dual:0} for this value of~$t$.

    In view of \eqref{eq:seller:Uaut} and \eqref{eq:seller:Vaut}, Lemma \ref{lem:2012-10-03:sigmai-intersect} consequently shows that
    \[
     \dom V^{\mathrm{ad}\mu}_t
     = \dom W^{\mathrm{ad}\mu}_t \cap \dom U^{\mathrm{ad}\mu}_t
     = \dom W^{\mathrm{ad}\mu}_t \cap \mathcal{Q}^{\mu\ast}_t
    \]
    is also compactly $\cc$-generated. We may apply Lemma~\ref{lem:Roux_Zastawniak2011} to the sets $-\mathcal{U}^{\mathrm{ad}\mu}_t$ and $-\mathcal{V}^{\mathrm{ad}\mu}_t$ since
    \[
     \mathcal{U}^{\mathrm{ad}\mu}_t \cap \mathcal{V}^{\mathrm{ad}\mu}_t = \mathcal{Z}^{\mathrm{ad}\mu}_t \neq \emptyset.
    \]
    Claim \eqref{item:prop:seller:dual:2} for this value of $t$ follows upon observing that
    \[
     \dom Z^{\mathrm{ad}\mu}_t = \conv\{\dom U^{\mathrm{ad}\mu}_t, \dom V^{\mathrm{ad}\mu}_t\} = \conv\{\mathcal{Q}^{\mu\ast}_t, \dom W^{\mathrm{ad}\mu}_t \cap \mathcal{Q}^{\mu\ast}_t\} = \mathcal{Q}^{\mu\ast}_t.
    \]
This concludes the inductive step.
\end{proof}

\begin{proof}[\textbf{\upshape Proof of Proposition~\ref{Prop:20130727:pi-ag-dual}.}]
  By Proposition~\ref{prop:seller:dual},  $\dom Z^\mathrm{ad}_0=\mathcal{Q}^\ast_0$ is compactly $\cc$-generated. Since~$-Z^\mathrm{ad}_0$ is polyhedral, it is continuous on its effective domain and therefore attains a maximum on the non-empty compact set $\sigmac(\mathcal{Q}^\ast_0)$. From Theorem~\ref{prop:seller:Zau0=initial-endowments} it follows \cite[Theorem~13.1]{rockafellar1996} that%
\begin{align*}
 \pi^\mathrm{ag}_\cc(\xi)
 &= \min\{x\in\mathbb{R}\st xe^\cc\in\mathcal{Z}^\mathrm{ad}_0\} \\
 &= \min\{x\in\mathbb{R}\st xe^\cc\cdot s \ge -Z^\mathrm{ad}_0(s)\text{ for all }s\in\mathbb{R}^d\} \\
 &= \min\{x\in\mathbb{R}\st xe^\cc\cdot s \ge -Z^\mathrm{ad}_0(s)\text{ for all }s\in\mathcal{Q}^\ast_0\} \\
 &= \min\{x\in\mathbb{R}\st x \ge -Z^\mathrm{ad}_0(s)\text{ for all }s\in\sigmac(\mathcal{Q}^\ast_0)\} \\
 &= \max\{-Z^\mathrm{ad}_0(s)\st s\in\sigmac(\mathcal{Q}^\ast_0)\}\\
 &= \max\{-Z^\mathrm{ad}_0(s)\st s\in\sigmac(\mathbb{R}^d)\}.
\end{align*}%

 The following construction produces adapted processes $y_t$, $\lambda_t$, $x_t$ and $S_t$ for $t=0,\ldots,T$, and $q_t$ for $t=1,\ldots,T$.  We already know that the maximum of $-Z^\mathrm{ad}_0$ over the set $\sigmac(\mathbb{R}^d)$ is attained, i.e.\ there exists some $y_0\in\sigmac(\mathcal{Q}^\ast_0)$ such that
\begin{equation}\label{eq:20130210-Z0}
 \pi^\mathrm{ag}_\cc(\xi)=-Z^\mathrm{ad}_0(y_0) = \max\{-Z^\mathrm{ad}_0(s)\st s\in\sigmac(\mathbb{R}^d)\}.
\end{equation}
  For any $t=0,\ldots,T-1$, suppose that $y_t\in\sigmac(\mathcal{Q}^\ast_t)$ is given, and fix any $\mu\in\Omega_t$. Then by Proposition~\ref{prop:seller:dual}\eqref{item:prop:seller:dual:2}, there exist $\lambda^\mu_t\in[0,1]$, $x^\mu_t\in\sigmac(\dom V_t^{\mathrm{ad}\mu})$ and $S^\mu_t\in\sigmac(\mathcal{Q}^{\mu\ast}_t)$ such that
 \begin{align}
  Z^{\mathrm{ad}\mu}_t(y^\mu_t) &= \lambda^\mu_tU^{\mathrm{ad}\mu}_t(S^\mu_t) + (1-\lambda^\mu_t)V^{\mathrm{ad}\mu}_t(x^\mu_t),  \label{eq:seller:dual-opt:1} \\
  y^\mu_t &= \lambda^\mu_t S^\mu_t + (1-\lambda^\mu_t)x^\mu_t. \label{eq:seller:dual-opt:2}
 \end{align}
By \eqref{eq:seller:Vaut} and Proposition~\ref{prop:seller:dual}\eqref{item:prop:seller:dual:3}, there exist $q^\nu_{t+1}\in[0,1]$ and $y^\nu_{t+1}\in\sigmac(\mathcal{Q}^{\nu\ast}_{t+1})$ for all $\nu\in\successors\mu$ such that
    \begin{align}
      V^\mathrm{ad\mu}_t(x^\mu_t) = W^\mathrm{ad\mu}_t(x^\mu_t) &= \sum_{\nu\in\successors\mu}q^\nu_{t+1} Z^\mathrm{ad\nu}_{t+1}(y^\nu_{t+1}), \label{eq:seller:dual-opt:3}\\
      x^\mu_t &= \sum_{\nu\in\successors\mu}q^\nu_{t+1} y^\nu_{t+1},  \label{eq:seller:dual-opt:4}\\
      1 &= \sum_{\nu\in\successors\mu}q^\nu_{t+1}.\nonumber
    \end{align}
This completes the inductive step.
Also define for all $\mu\in\Omega_T$
 \begin{align*}
  \lambda^\mu_T &:= 1, & x^\mu_T &:= S^\mu_T := y^\mu_T.
 \end{align*}
 Then \eqref{eq:seller:dual-opt:1}, \eqref{eq:seller:dual-opt:2} are also satisfied when $t=T$.

 Define the probability measure $\mathbb{Q}$ on $\mathcal{F}_T=2^\Omega$ as
 \[
  \mathbb{Q}(\{\omega\}) := \prod_{t=1}^{T}q^{\omega\uparrow t}_{t} \text{ for all }\omega\in\Omega,
 \]
 where $\omega\uparrow t$ denotes the element of $\Omega_t$ that contains $\omega$. It then follows from \eqref{eq:seller:dual-opt:3}, \eqref{eq:seller:dual-opt:4} that for all $t=0,\ldots,T-1$
 \begin{align}
   V^\mathrm{ad}_t(x_t) &= \mathbb{E}_\mathbb{Q}(Z^\mathrm{ad}_{t+1}(y_{t+1})|\mathcal{F}_t),\label{eq:seller:dual-opt:8}\\
 x_t &= \mathbb{E}_\mathbb{Q}(y_{t+1}|\mathcal{F}_t).\label{eq:seller:dual-opt:9}
 \end{align}

The mixed stopping time $\chi$ is defined by setting $\chi_0:=\lambda_0$ and
\[
 \chi_t := \lambda_t\left(1-\sum_{s=0}^{t-1}\chi_s\right) \text{ for all }t=1,\ldots,T.
\]
It is straightforward to show by induction that $\chi_t\ge0$ for all $t$. Moreover, since $\lambda_T=1$, we have
\[
 \sum_{t=0}^T\chi_t=1.
\]
Observe also that
\begin{align*}
 \lambda_t\chi^\ast_t &= \chi_t, & (1-\lambda_t)\chi^\ast_t = \chi^\ast_t - \chi_t = \chi^\ast_{t+1}
\end{align*}
for all $t=0,\ldots,T$, where $\chi^\ast$ is defined by~(\ref{eq:2013-07-13-chi-star}). It then follows from \eqref{eq:seller:dual-opt:1}, \eqref{eq:seller:dual-opt:2} and \eqref{eq:seller:Uaut} that for all $t=0,\ldots,T$
\begin{align}
 \chi^\ast_tZ^\mathrm{ad}_t(y_t) &= -\chi_t\xi_t\cdot S_t + \chi^\ast_{t+1}V^\mathrm{ad}_t(x_t),\label{eq:seller:dual-opt:6}\\
 \chi^\ast_ty_t &= \chi_tS_t + \chi^\ast_{t+1}x_t.\label{eq:seller:dual-opt:7}
\end{align}

We now show by backward induction that for all $t=0,\ldots,T$
\begin{equation}
 \mathbb{E}_\mathbb{Q}(S^{\chi\ast}_{t+1}|\mathcal{F}_t) = \chi^\ast_{t+1}x_t.  \label{eq:seller:dual-opt:5}
\end{equation}
At time $T$ the result is trivial because $S^{\chi\ast}_{T+1} = \chi^\ast_{T+1}=0$. Suppose now that (\ref{eq:seller:dual-opt:5}) holds for some $t=1,\ldots,T$.
Then, by the tower property of conditional expectation,
\begin{align*}
 \mathbb{E}_\mathbb{Q}(S^{\chi\ast}_t|\mathcal{F}_{t-1})
 &= \mathbb{E}_\mathbb{Q}(\chi_tS_t + S^{\chi\ast}_{t+1}|\mathcal{F}_{t-1}) \\
 &= \mathbb{E}_\mathbb{Q}(\chi_tS_t + \mathbb{E}_\mathbb{Q}(S^{\chi\ast}_{t+1}|\mathcal{F}_t)|\mathcal{F}_{t-1}) \\
 &= \mathbb{E}_\mathbb{Q}(\chi_tS_t + \chi^\ast_{t+1}x_t|\mathcal{F}_{t-1})
\end{align*}
and, by \eqref{eq:seller:dual-opt:7}, the predictability of $\chi^\ast$, and \eqref{eq:seller:dual-opt:9},
\begin{align*}
 \mathbb{E}_\mathbb{Q}(S^{\chi\ast}_t|\mathcal{F}_{t-1}) = \mathbb{E}_\mathbb{Q}(\chi^\ast_ty_t|\mathcal{F}_{t-1}) = \chi^\ast_t\mathbb{E}_\mathbb{Q}(y_t|\mathcal{F}_{t-1}) = \chi^\ast_tx_{t-1}.
\end{align*}
This concludes the inductive step.

We also show by backward induction that for all $t=0,\ldots,T$
\begin{equation}\label{eq:seller:dual-opt:10}
 \chi^\ast_t Z^\mathrm{ad}_t(y_t) = -\mathbb{E}_\mathbb{Q}((\xi\cdot S)^{\chi\ast}_t|\mathcal{F}_t).
\end{equation}
At time~$T$
\[
 \chi^\ast_T Z^\mathrm{ad}_T(y_T) = \chi^\ast_T U^\mathrm{ad}_T(y_T) = -\chi_T\xi_T\cdot S_T = - (\xi\cdot S)^{\chi\ast}_T = -\mathbb{E}_\mathbb{Q}((\xi\cdot S)^{\chi\ast}_T|\mathcal{F}_T).
\]
Suppose now that (\ref{eq:seller:dual-opt:10}) holds for some $t=1,\ldots,T$.
Then by \eqref{eq:seller:dual-opt:6}, \eqref{eq:seller:dual-opt:8} and the tower property of conditional expectation, we have
\begin{align*}
 \chi^\ast_{t-1} Z^\mathrm{ad}_{t-1}(y_{t-1})
 &=-\chi_{t-1}\xi_{t-1}\cdot S_{t-1} + \chi^\ast_{t}V^\mathrm{ad}_{t-1}(x_{t-1})\\
 &=-\chi_{t-1}\xi_{t-1}\cdot S_{t-1} + \chi^\ast_{t}\mathbb{E}_\mathbb{Q}(Z^\mathrm{ad}_{t}(y_{t})|\mathcal{F}_{t-1}) \\
 &=-\chi_{t-1}\xi_{t-1}\cdot S_{t-1} + \mathbb{E}_\mathbb{Q}(\chi^\ast_{t}Z^\mathrm{ad}_{t}(y_{t})|\mathcal{F}_{t-1}) \\
 &=-\chi_{t-1}\xi_{t-1}\cdot S_{t-1} -\mathbb{E}_\mathbb{Q}(\mathbb{E}_\mathbb{Q}((\xi\cdot S)^{\chi\ast}_{t}|\mathcal{F}_{t})|\mathcal{F}_{t-1}) \\
 &=-\chi_{t-1}\xi_{t-1}\cdot S_{t-1} -\mathbb{E}_\mathbb{Q}((\xi\cdot S)^{\chi\ast}_{t}|\mathcal{F}_{t-1}) \\
 &=- \mathbb{E}_\mathbb{Q}((\xi\cdot S)^{\chi\ast}_{t-1}|\mathcal{F}_{t-1}).
\end{align*}
This concludes the inductive step.

Since $x_t\in\dom V_t^\mathrm{ad}\subseteq\mathcal{Q}^\ast_t$ for all $t$, property \eqref{eq:seller:dual-opt:5} implies that $\mathbb{E}_\mathbb{Q}(S^{\chi\ast}_{t+1}|\mathcal{F}_t)\in\mathcal{Q}^\ast_t$ for all $t=0,\ldots,T$. Moreover, by~\eqref{eq:seller:dual-opt:10},
\[
 -Z^\mathrm{ad}_0(y_0)=-\chi^\ast_0 Z^\mathrm{ad}_0(y_0) = \mathbb{E}_\mathbb{Q}((\xi\cdot S)^{\chi\ast}_0|\mathcal{F}_0) = \mathbb{E}_\mathbb{Q}((\xi\cdot S)\chi).
\]
From (\ref{eq:20130210-Z0}) we therefore have $\pi^\mathrm{ag}_\cc(\xi)=\mathbb{E}_\mathbb{Q}((\xi\cdot S)_\chi)$, as required.
\end{proof}

\begin{proof}[\textbf{\upshape Proof of Theorem~\ref{Thm:ask-price-representation}.}]
By Proposition \ref{Prop:20130727:pi-ag-dual}, a stopping time $\hat{\chi}\in\mathcal{X}$ and a pair $(\hat{\mathbb{Q}},\hat{S})\in\Pbaruc(\hat{\chi})$ can be constructed such that
\[
 \pi^\mathrm{ag}_\cc(\xi) =
 \mathbb{E}_{\hat{\mathbb{Q}}}((\xi\cdot \hat{S})_{\hat{\chi}}) \le \max_{\chi\in\mathcal{X}}\max_{(\mathbb{Q},S)\in\Pbaruc(\chi)}\mathbb{E}_\mathbb{Q}((\xi\cdot S)_\chi).
\]
To establish the reverse inequality we prove by backward induction that for any $y\in\Phi^\mathrm{ad}(\xi)$, $\chi\in\mathcal{X}$ and $(\mathbb{Q},S)\in\bar{\mathcal{P}}^\mathrm{d}(\chi)$
 \begin{equation}\label{eq:reverse-ineq-dual repr-seller}
  y_t \cdot \mathbb{E}_\mathbb{Q}(S^{\chi\ast}_t|\mathcal{F}_t) \ge \mathbb{E}_\mathbb{Q}((\xi\cdot S)^{\chi\ast}_t|\mathcal{F}_t)\text{ for all }t=0,\ldots,T.
 \end{equation}
When $t=T$,
 \[
  y^{}_T\cdot \mathbb{E}_\mathbb{Q}(S^{\chi\ast}_T|\mathcal{F}^{}_T) = \chi^{}_Ty^{}_T\cdot S^{}_T \ge \chi^{}_T\xi^{}_T\cdot S^{}_T = \mathbb{E}_\mathbb{Q}((\xi\cdot S)^{\chi\ast}_T|\mathcal{F}^{}_T)
 \]
 since $y_T-\xi_T\in\mathcal{Q}_T$ and $S_T\in\mathcal{Q}^\ast_T$.
 Now fix any $t=0,\ldots,T-1$, and suppose that
 \[
  y_{t+1} \cdot \mathbb{E}_\mathbb{Q}(S^{\chi\ast}_{t+1}|\mathcal{F}_{t+1}) \ge \mathbb{E}_\mathbb{Q}((\xi\cdot S)^{\chi\ast}_{t+1}|\mathcal{F}_{t+1}).
 \]
 Then, by the tower property of conditional expectation, and since $y_t-y_{t+1}\in\mathcal{Q}_t$ and $\mathbb{E}_\mathbb{Q}(S^{\chi\ast}_t|\mathcal{F}_t)\in\mathcal{Q}^\ast_t$, it follows that
 \begin{align*}
  y_t\cdot\mathbb{E}_\mathbb{Q}(S^{\chi\ast}_t|\mathcal{F}_t)
  &= \chi_ty_t\cdot S_t + y_t\cdot\mathbb{E}_\mathbb{Q}(S^{\chi\ast}_{t+1}|\mathcal{F}_t) \\
  &\ge \chi_t\xi_t\cdot S_t + y_{t+1}\cdot \mathbb{E}_\mathbb{Q}(S^{\chi\ast}_{t+1}|\mathcal{F}_t) \\
  &=\chi_t\xi_t\cdot S_t + \mathbb{E}_\mathbb{Q}( y_{t+1}\cdot \mathbb{E}_\mathbb{Q}(S^{\chi\ast}_{t+1}|\mathcal{F}_{t+1})|\mathcal{F}_t) \\
  &\ge\chi_t\xi_t\cdot S_t + \mathbb{E}_\mathbb{Q}( \mathbb{E}_\mathbb{Q}((\xi\cdot S)^{\chi\ast}_{t+1}|\mathcal{F}_{t+1})|\mathcal{F}_t) \\
  &=\chi_t\xi_t\cdot S_t + \mathbb{E}_\mathbb{Q}((\xi\cdot S)^{\chi\ast}_{t+1}|\mathcal{F}_t) = \mathbb{E}_\mathbb{Q}((\xi\cdot S)^{\chi\ast}_t|\mathcal{F}_t),
 \end{align*}
which proves~(\ref{eq:reverse-ineq-dual repr-seller}).
The construction in the proof of Theorem~\ref{prop:seller:Zau0=initial-endowments} with initial portfolio $\hat{y}_0=\pi^\mathrm{ag}_\cc(\xi)e^\cc\in\mathcal{Z}^\mathrm{ad}_0$ yields a strategy $\hat{y}\in\Phi^\mathrm{ad}(\xi)$. For any $\chi\in\mathcal{X}$ and $(\mathbb{Q},S)\in\Pbaruc(\chi)$ we have $\mathbb{E}_\mathbb{Q}(S^\cc_\chi)=1$, and therefore (\ref{eq:reverse-ineq-dual repr-seller}) with $t=0$ yields
\[
 \pi^\mathrm{ag}_\cc(\xi) = \hat{y}_0 \cdot \mathbb{E}_\mathbb{Q}(S_\chi)\ge \mathbb{E}_\mathbb{Q}((\xi\cdot S)_\chi).
\]
It follows that
\[
 \pi^\mathrm{ag}_\cc(\xi) \ge \max_{\chi\in\mathcal{X}}\max_{(\mathbb{Q},S)\in\Pbaruc(\chi)}\mathbb{E}_\mathbb{Q}((\xi\cdot S)_\chi).
\]
\end{proof}

\subsubsection{Buyer's case}\label{Sect:Appendix:technical:buyer}

 \begin{proof}[\textbf{\upshape Proof of Proposition~\ref{prop:2012-09-19:zt-closed}.}]
   As $\mathcal{Q}_t$ and $\mathcal{L}_t$ are polyhedral cones, they are closed and convex. We have $0^+\mathcal{Q}_t=\mathcal{Q}_t$ and $0^+\mathcal{L}_t=\mathcal{L}_t$. It follows that $0^+\mathcal{U}^\mathrm{bd}_t=\mathcal{Q}_t$ for all $t=0,\ldots,T$ \citep[Corollary~9.1.2]{rockafellar1996}.

   The set $\mathcal{Z}^\mathrm{bd}_T = \mathcal{U}^\mathrm{bd}_T=-\xi_T+\mathcal{Q}_T$ is clearly polyhedral with recession cone~$\mathcal{Q}_T$. For $t=0,\ldots,T-1$ we proceed by induction. Suppose that $\mathcal{Z}^\mathrm{bd}_{t+1}$ is polyhedral and its recession cone is $\mathcal{Q}_{t+1}$. Then $\mathcal{W}^\mathrm{bd}_t=\mathcal{Z}^\mathrm{bd}_{t+1}\cap\mathcal{L}_t$ is polyhedral and its recession cone is $\mathcal{Q}_{t+1}\cap\mathcal{L}_t$ \citep[Corollary~8.3.3]{rockafellar1996}. Being polyhedral, $\mathcal{W}^\mathrm{bd}_t$ is the convex hull of a finite set of points and directions, and its recession cone $\mathcal{Q}_{t+1}\cap\mathcal{L}_t$ is the convex hull of the origin and the directions in~$\mathcal{W}^\mathrm{bd}_t$.

   The set $\mathcal{V}^\mathrm{bd}_t=\mathcal{W}^\mathrm{bd}_t+\mathcal{Q}_t$ is polyhedral \citep[Corollary~19.3.2]{rockafellar1996} and hence it is the convex hull of a finite set of points and directions. Since the cone $\mathcal{Q}_t$ can be written as the convex hull of the origin and a finite number of directions, it is possible to write $\mathcal{V}^\mathrm{bd}_t$ as the convex hull of a finite set of points, all in $\mathcal{W}^\mathrm{bd}_t$, and a finite set of directions. These directions are exactly the directions in $\mathcal{W}^\mathrm{bd}_t$ and~$\mathcal{Q}_t$, i.e.~the directions in $\mathcal{Q}_{t+1}\cap\mathcal{L}_t$ and $\mathcal{Q}_t$. Thus the recession cone of $\mathcal{V}^\mathrm{bd}_t$ is
    \[
    0^+\mathcal{V}^\mathrm{bd}_t=\conv\{\mathcal{Q}_t,\mathcal{Q}_{t+1}\cap\mathcal{L}_t\}=\mathcal{Q}_t
    \]
    since $\mathcal{Q}_{t+1}\cap\mathcal{L}_t\subseteq\mathcal{Q}_{t}$ by~(\ref{eq:Qt-recursive}).
    This means that the set $\mathcal{Z}^\mathrm{bd}_t=\conv\{\mathcal{V}^\mathrm{bd}_t,\mathcal{U}^\mathrm{bd}_t\}$ is closed and its recession cone is $\mathcal{Q}_t$ \citep[Corollary~9.8.1]{rockafellar1996}. Moreover, since $\mathcal{V}^\mathrm{bd}_t$ and $\mathcal{U}^\mathrm{bd}_t$ are polyhedral, it follows that $\mathrm{cl}\mathcal{Z}^\mathrm{bd}_t=\mathrm{cl}(\conv\{\mathcal{V}^\mathrm{bd}_t,\mathcal{U}^\mathrm{bd}_t\})$ is polyhedral  \citep[Theorem~19.6]{rockafellar1996}, which means that $\mathcal{Z}^\mathrm{bd}_t=\mathrm{cl}\mathcal{Z}^\mathrm{bd}_t$ is polyhedral, concluding the inductive step.
 \end{proof}

\begin{proposition}\label{prop:2012-07-26:hedging-construct-converse}
 If $(y,\chi)\in\Phi^\mathrm{bg}(\xi)$, then for all $t=0,\ldots,T$
 \[
  y_t \in
  \begin{cases}
   \chi^\ast_t\mathcal{Z}^\mathrm{bd}_t & \text{on }\{\chi^\ast_t>0\},\\
   \mathcal{Q}_t & \text{on }\{\chi^\ast_t=0\}.
  \end{cases}
 \]
\end{proposition}

\begin{proof}
 The proof is by backward induction. Since $y_{T+1}=0$, from~\eqref{eq:buyer-self-fin-superhedge} we have
 \[
  y_T \in -\chi_T\xi_T + \mathcal{K}_T \subseteq -\chi_T\xi_T + \mathcal{Q}_T.
 \]
 It immediately follows that $y_T\in\mathcal{Q}_T$ on the set $\{\chi^\ast_T=0\}=\{\chi_T=0\}$. On the set $\{\chi^\ast_T>0\}=\{\chi_T>0\}$ we have $\mathcal{Q}_T=\chi_T\mathcal{Q}_T$ because $\mathcal{Q}_T$ is a cone, and therefore
 \[
  y_T \in \chi_T(-\xi_T + \mathcal{Q}_T) = \chi_T\mathcal{U}^\mathrm{bd}_T = \chi^\ast_T\mathcal{Z}^\mathrm{bd}_T.
 \]

 Suppose now for some $t=0,\ldots,T-1$ that
  \[
  y_{t+1} \in
  \begin{cases}
   \chi^\ast_{t+1}\mathcal{Z}^\mathrm{bd}_{t+1} & \text{on }\{\chi^\ast_{t+1}>0\},\\
   \mathcal{Q}_{t+1} & \text{on }\{\chi^\ast_{t+1}=0\}.
  \end{cases}
 \]
 Because $y_{t+1},\chi^\ast_{t+1}\in\mathcal{L}_t$, this means that
   \[
  y_{t+1} \in
  \begin{cases}
   \chi^\ast_{t+1}\mathcal{W}^\mathrm{bd}_t & \text{on }\{\chi^\ast_{t+1}>0\},\\
   \mathcal{Q}_{t+1}\cap\mathcal{L}_t & \text{on }\{\chi^\ast_{t+1}=0\}.
  \end{cases}
 \]
Since $y_t\in-\chi_t\xi_t+y_{t+1}+\mathcal{K}_t\subseteq -\chi_t\xi_t+y_{t+1}+\mathcal{Q}_t$ by~\eqref{eq:buyer-self-fin-superhedge}, it follows that
   \[
  y_t \in
  \begin{cases}
   -\chi_t\xi_t+\chi^\ast_{t+1}\mathcal{W}^\mathrm{bd}_t+\mathcal{Q}_t & \text{on }\{\chi^\ast_{t+1}>0\},\\
   -\chi_t\xi_t+\mathcal{Q}_{t+1}\cap\mathcal{L}_t+\mathcal{Q}_t & \text{on }\{\chi^\ast_{t+1}=0\}.
  \end{cases}
 \]
 We consider the two possibilities separately.
 \begin{itemize}
  \item On the set $\{\chi^\ast_{t+1}>0\}$ we have $\chi^\ast_t>0$ and therefore
  \[
   \mathcal{Q}_t = \chi^\ast_t\mathcal{Q}_t \subseteq \chi_t\mathcal{Q}_t + \chi^\ast_{t+1}\mathcal{Q}_t,
  \]
  so that
  \begin{multline*}
   y_t\in-\chi_t\xi_t+\chi^\ast_{t+1}\mathcal{W}^\mathrm{bd}_t+\chi_t\mathcal{Q}_t + \chi^\ast_{t+1}\mathcal{Q}_t \\= \chi_t(-\xi_t+\mathcal{Q}_t)+\chi^\ast_{t+1}(\mathcal{W}^\mathrm{bd}_t+\mathcal{Q}_t) = \chi_t\mathcal{U}^\mathrm{bd}_t + \chi^\ast_{t+1}\mathcal{V}^\mathrm{bd}_t.
  \end{multline*}
  Since
  \[
   \chi_t\mathcal{U}^\mathrm{bd}_t + \chi^\ast_{t+1}\mathcal{V}^\mathrm{bd}_t = \chi^\ast_t\left(\frac{\chi_t}{\chi^\ast_t}\mathcal{U}^\mathrm{bd}_t + \frac{\chi^\ast_{t+1}}{\chi^\ast_t}\mathcal{V}^\mathrm{bd}_t\right) \subseteq \chi^\ast_t\mathcal{Z}^\mathrm{bd}_t,
  \]
  it follows that $y_t\in\chi^\ast_t\mathcal{Z}^\mathrm{bd}_t$ on $\{\chi^\ast_{t+1}>0\}$.

  \item On the set $\{\chi^\ast_{t+1}=0\}$ we have
  \[
   y_t \in -\chi_t\xi_t+\mathcal{Q}_{t+1}\cap\mathcal{L}_t+\mathcal{Q}_t \subseteq -\chi_t\xi_t+\mathcal{Q}_t
  \]
  because $\mathcal{Q}_{t+1}\cap\mathcal{L}_t\subseteq\mathcal{Q}_t$ by~\eqref{eq:Qt-recursive}. There are two further possibilities.
  \begin{itemize}
    \item On $\{\chi^\ast_t>0\}\cap\{\chi^\ast_{t+1}=0\}$ we have $\chi_t=\chi^\ast_t>0$ and therefore
  \[
   y_t\in-\chi_t\xi_t+\mathcal{Q}_t=\chi_t(-\xi_t+\mathcal{Q}_t)=\chi_t\mathcal{U}^\mathrm{bd}_t\subseteq\chi_t\mathcal{Z}^\mathrm{bd}_t=\chi^\ast_t\mathcal{Z}^\mathrm{bd}_t.
  \]

  \item On $\{\chi^\ast_t=0\}\subseteq\{\chi^\ast_{t+1}=0\}$ we have $\chi_t=0$ and therefore $y_t\in\mathcal{Q}_t$ as claimed.
  \end{itemize}
 \end{itemize}
\end{proof}

\begin{proof}[\textbf{\upshape Proof of Theorem~\ref{prop:2012-07-26:hedging-construct}.}]
In view of Proposition~\ref{prop:2012-07-26:hedging-construct-converse}, to verify~(\ref{eq:2012-07-26:constr-equivalence}) it is
sufficient to show that for every $x\in\mathcal{Z}_{0}^{\mathrm{bd}}$ there
exists a pair $(y,\chi)\in\Phi^{\mathrm{bg}}(\xi)$ such that $y_{0}=x$. To
this end, define $p_{0}:=1$ and $z_{0}:=x\in p_{0}\mathcal{Z}_{0}%
^{\mathrm{bd}}$. Suppose by induction that for some $t=0,\ldots,T-1$ we have
constructed predictable sequences $z_{0},\ldots,z_{t}$ and $p_{0},\ldots
,p_{t}$ such that $1= p_{0}\geq\cdots\geq p_{t}\geq0$ and
\begin{align*}
z_{s}+(p_{s}-p_{s+1})\xi_{s}-z_{s+1} &  \in\mathcal{Q}_{s} & \text{ for all }
&  s=0,\ldots,t-1,\\
z_{s} &  \in p_{s}\mathcal{Z}_{s}^{\mathrm{bd}} & \text{ for all } &
s=0,\ldots,t.
\end{align*}
Because of \eqref{eq:mathcal-Zt}, there exists an $\mathcal{F}_{t}$-measurable
random variable $\lambda_{t}$ such that $0\leq\lambda_{t}\leq1$ and
\[
z_{t}\in(1-\lambda_{t})p_{t}\mathcal{U}_{t}^{\mathrm{bd}}+\lambda_{t}%
p_{t}\mathcal{V}_{t}^{\mathrm{bd}}.
\]
Equations \eqref{eq:mathcal-Ut} and \eqref{eq:mathcal-Vt} then give
\[
z_{t}\in(1-\lambda_{t})p_{t}(-\xi_{t}+\mathcal{Q}_{t})+\lambda_{t}%
p_{t}(\mathcal{W}_{t}^{\mathrm{bd}}+\mathcal{Q}_{t})=-(1-\lambda_{t})p_{t}%
\xi_{t}+p_{t}\mathcal{Q}_{t}+\lambda_{t}p_{t}\mathcal{W}_{t}^{\mathrm{bd}},
\]
where
\[
(1-\lambda_{t})p_{t}\mathcal{Q}_{t}+\lambda_{t}p_{t}\mathcal{Q}_{t}%
=p_{t}\mathcal{Q}_{t}%
\]
follows from the fact that $\mathcal{Q}_{t}$ is a convex cone. This means
there exists a random variable
\[
z_{t+1}\in\lambda_{t}p_{t}\mathcal{W}_{t}^{\mathrm{bd}}=\lambda_{t}%
p_{t}(\mathcal{Z}_{t+1}^{\mathrm{bd}}\cap\mathcal{L}_{t})=(\lambda_{t}%
p_{t}\mathcal{Z}_{t+1}^{\mathrm{bd}})\cap\mathcal{L}_{t}%
\]
such that
\[
z_{t}+(1-\lambda_{t})p_{t}\xi_{t}-z_{t+1}\in p_{t}\mathcal{Q}_{t}%
\subseteq\mathcal{Q}_{t}.
\]
Put $p_{t+1}:=\lambda_{t}p_{t}$. Then $z_{t+1}\in p_{t+1}\mathcal{Z}%
_{t+1}^{\mathrm{bd}}$, which concludes the inductive step. Now define the
mixed stopping time $\chi=(\chi_{t})$ by
\[
\chi_{t}:=%
\begin{cases}
p_{T} & \text{if }t=T,\\
p_{t}-p_{t+1} & \text{if }t=0,\ldots,T-1.
\end{cases}
\]
We also put $z_{T+1}:=0$. We have constructed $z\in\Phi$ and $\chi\in\mathcal{X}$ such that
$z_{0}=x$ and%
\[
z_{t}+\chi_{t}\xi_{t}-z_{t+1}\in\mathcal{Q}_{t}\text{ for each }t=0,\ldots,T.
\]
Finally, we construct $y\in\Phi$ such that $(y,\chi)\in\Phi^{\mathrm{bg}}%
(\xi)$\ and $y_{0}=x$. By the definition of the deferred solvency
cones~$\mathcal{Q}_{t}$, for each $t=0,\ldots,T$ there is a liquidation strategy
$v_{t+1}^{t},\ldots,v_{T+1}^{t}$ starting from $z_{t}+\chi_{t}\xi_{t}-z_{t+1}$ at time~$t$.
We put%
\[
y_{t}:=z_{t}+\sum_{r=0}^{t-1}\sum_{s=t}^{T}\left(  v_{s}^{r}-v_{s+1}%
^{r}\right)  ,
\]
which means that%
\begin{align*}
y_{t}+\chi_{t}\xi_{t}-y_{t+1} &  =z_{t}+\sum_{r=0}^{t-1}\sum_{s=t}^{T}\left(
v_{s}^{r}-v_{s+1}^{r}\right)  +\chi_{t}\xi_{t}-z_{t+1}-\sum_{r=0}^{t}%
\sum_{s=t+1}^{T}\left(  v_{s}^{r}-v_{s+1}^{r}\right)  \\
&  =z_{t}+\chi_{t}\xi_{t}-z_{t+1}-v_{t+1}^{t}+\sum_{r=0}^{t-1}\left(
v_{t}^{r}-v_{t+1}^{r}\right)  \in\mathcal{K}_{t}%
\end{align*}
for each $t=0,\ldots,T$, with $y_{0}=z_{0}=x$, completing the proof of~(\ref{eq:2012-07-26:constr-equivalence}).

Next, if follows from~(\ref{eq:2012-07-26:constr-equivalence}) that
\begin{align*}
\pi_{\cc}^{\mathrm{bg}}(\xi)
& =\sup\left\{  -x\in\mathbb{R}\st
\exists (y,\chi)\in\Phi^\mathrm{bg}(\xi): xe^{\cc}=y_{0}\right\}  \\
& =\sup\left\{  -x\in\mathbb{R}\st xe^{\cc}\in\mathcal{Z}_{0}^{\mathrm{bd}%
}\right\}  .
\end{align*}
By Proposition~\ref{prop:2012-09-19:zt-closed}, $\mathcal{Z}_{0}^{\mathrm{bd}}$ is
polyhedral, hence closed. As a result, the set $\left\{  -x\in\mathbb{R}\st
xe^{\cc}\in\mathcal{Z}_{0}^{\mathrm{bd}}\right\}$ is also closed. It is non-empty and bounded above because $xe^\cc\in\mathcal{Z}^\mathrm{bd}_0$ for any $x\in\mathbb{R}$ large enough, and $xe^\cc\notin\mathcal{Z}^\mathrm{bd}_0$ for any $x\in\mathbb{R}$ small enough. This means that the supremum is attained. It follows that $-\pi_{\cc}^{\mathrm{bg}}(\xi)e^\cc\in\mathcal{Z}^\mathrm{bd}_0$, so we know that a strategy $(y,\chi)\in\Phi^\mathrm{bg}(\xi)$ can be constructed such that $\pi_{\cc}^{\mathrm{bg}}(\xi)e^\cc=-y_0$.
\end{proof}

\begin{proof}[\textbf{\upshape Proof of Proposition~\ref{Prop:Am-Eur}.}]
For any $y\in\Phi$ such that $\left(  y,\chi\right)  \in\Phi^{\mathrm{bg}}%
(\xi)$, put%
\[
z_{t}:=y_{t}-\sum_{s=0}^{t-1}\chi_{s}\xi_{s}%
\]
for each $t=0,\ldots,T$, and $z_{T+1}:=0$. Then%
\begin{align*}
z_{t}-z_{t+1} &  =y_{t}-\sum_{s=0}^{t-1}\chi_{s}\xi_{s}-y_{t+1}+\sum_{s=0}%
^{t}\chi_{s}\xi_{s}\\
&  =y_{t}+\chi_{t}\xi_{t}-y_{t+1}\in\mathcal{K}_{t}%
\end{align*}
for each $t=0,\ldots,T-1$, and%
\begin{align*}
z_{T}+\xi_\chi &  =y_{T}-\sum_{s=0}^{T-1}\chi_{s}\xi
_{s}+\sum_{s=0}^{T}\chi_{s}\xi_{s}\\
&  =y_{T}+\chi_{T}\xi_{T}-y_{T+1}\in\mathcal{K}_{T}%
\end{align*}
since $y_{T+1}=0$, so $z\in\Psi^{\mathrm{a}}(-\xi_\chi)$ with $z_{0}=y_{0}$.

Conversely, for any $z\in\Psi^{\mathrm{a}}(-\xi_\chi)$ we put%
\[
y_{t}:=z_{t}+\sum_{s=0}^{t-1}\chi_{s}\xi_{s}%
\]
for each $t=0,\ldots,T$, and $y_{T+1}:=0$. Then%
\begin{align*}
y_{t}+\chi_{t}\xi_{t}-y_{t+1} &  =z_{t}+\sum_{s=0}^{t-1}\chi_{s}\xi_{s}%
+\chi_{t}\xi_{t}-z_{t+1}-\sum_{s=0}^{t}\chi_{s}\xi_{s}\\
&  =z_{t}-z_{t+1}\in\mathcal{K}_{t}%
\end{align*}
for each $t=0,\ldots,T-1$, and%
\begin{align*}
y_{T}+\chi_{T}\xi_{T}-y_{T+1} &  =z_{T}+\sum_{s=0}^{T-1}\chi_{s}\xi_{s}%
+\chi_{T}\xi_{T}\\
&  =z_{t}+\sum_{s=0}^{T}\chi_{s}\xi_{s}\in\mathcal{K}_{T}.
\end{align*}
It follows that $\left(  y,\chi\right)  \in\Phi^{\mathrm{bg}}(\xi)$ and
$y_{0}=z_{0}$, completing the proof.
\end{proof}

\begin{proof}[\textbf{\upshape Proof of Theorem~\ref{th:bu-buyer}.}]
Theorem~\ref{prop:2012-07-26:hedging-construct} gives
\[
 \pi^\mathrm{bg}_\cc(\xi) = \max\{-x\in\mathbb{R}\st xe^\cc\in\mathcal{Z}^\mathrm{bd}_0\}.
\]
The maximum is attained, so $-\pi^\mathrm{bg}_\cc(\xi)e^\cc\in\mathcal{Z}^\mathrm{bd}_0$. The strategy $(\hat{y},\hat{\chi})\in\Phi^\mathrm{bg}(\xi)$ constructed by the method in the proof of Theorem~\ref{prop:2012-07-26:hedging-construct} from the initial portfolio $\hat{y}_0=-\pi^\mathrm{bg}_\cc(\xi)e^\cc$ therefore realises the supremum in~(\ref{eq:buyer-bid-price-gradual}). We write this supremum as a maximum,
\begin{align*}
 \pi^\mathrm{bg}_\cc(\xi)
 &= \max\{-x\in\mathbb{R}\st \exists (y,\chi)\in\Phi^\mathrm{bg}(\xi): xe^\cc=y_0\}\\
 &= \max_{\chi\in\mathcal{X}}
 \left[\max\{-x\in\mathbb{R}\st \exists y\in\Phi: (y,\chi)\in\Phi^\mathrm{bg}(\xi), xe^\cc=y_0\}\right],
\end{align*}
and apply Proposition~\ref{Prop:Am-Eur}, which gives
\begin{align*}
 \pi^\mathrm{bg}_\cc(\xi)
 &= \max_{\chi\in\mathcal{X}}
 \left[\max\{-x\in\mathbb{R}\st \exists y\in\Psi^\mathrm{a}(-\xi_\chi): xe^\cc=y_0,\}\right]\\
  &= \max_{\chi\in\mathcal{X}}
 \left[-p^\mathrm{a}_\cc(-\xi_\chi)\right],
\end{align*}
where $p^\mathrm{a}_\cc(-\xi_\chi)$ is the ask (seller's) price in currency~$j$ of a European option with expiry time $T$ and payoff $-\xi_\chi\in\mathcal{L}_{T}$ as defined in Appendix~\ref{Sect:Eur-opt}.
We can now apply Lemma~\ref{Lem:Eur-ask-price-dual-repr} to write
\[
-p^\mathrm{a}_\cc(-\xi_\chi)
=-\max_{(\mathbb{Q},S)\in\Pbarc}\mathbb{E}_\mathbb{Q}\left(  -\xi_{\chi}\cdot S_{T}\right)
=\min_{(\mathbb{Q},S)\in\Pbarc}\mathbb{E}_\mathbb{Q}\left(  \xi_{\chi}\cdot S_{T}\right).
\]
For any $(\mathbb{Q},S)\in\Pbarc$, since $S$ is a martingale under~$\mathbb{Q}$, we have
\[
\mathbb{E}_{\mathbb{Q}}\left(\xi_{\chi}\cdot S_{T}\right)
=\mathbb{E}_{\mathbb{Q}}\left(\textstyle\sum_{t=0}^{T}\chi_{t}\xi_{t}\cdot
S_{T}\right)
 =\mathbb{E}_{\mathbb{Q}}\left(\textstyle\sum_{t=0}^{T}\chi_{t}\xi
_{t}\cdot S_{t}\right)
 =\mathbb{E}_{\mathbb{Q}}((\xi\cdot S)_{\chi}).
\]
This means that
\[
\pi^\mathrm{bg}_\cc(\xi)
= \max_{\chi\in\mathcal{X}}\min_{(\mathbb{Q},S)\in\Pbarc}
\mathbb{E}_{\mathbb{Q}}((\xi\cdot S)_{\chi}),
\]
proving~(\ref{eq:pi-bu}). We know that $(\hat{y},\hat{\chi})$ realises the supremum in~(\ref{eq:buyer-bid-price-gradual}), and therefore the above maxima over $\chi\in\mathcal{X}$ are attained at~$\hat{\chi}$. A pair $(\hat{\mathbb{Q}},\hat{S})\in\Pc$ such that
\[
\pi^\mathrm{bg}_\cc(\xi)
=-p^\mathrm{a}_\cc(-\xi_{\hat{\chi}})
= \mathbb{E}_{\hat{\mathbb{Q}}}((\xi\cdot \hat{S})_{\hat{\chi}}) =\displaystyle\min_{(\mathbb{Q},S)\in\Pbarc}\mathbb{E}_\mathbb{Q}\left((\xi\cdot S)_{\hat{\chi}}\right)
\]
can be constructed by the method of \citet[Proposition~5.3]{Roux_Zastawniak2011} for the European option with payoff $-\xi_{\hat{\chi}}$, completing the proof.
\end{proof}

\subsection{European options}\label{Sect:Eur-opt}

We recall a result for European options in the market model with $d$
assets under transaction costs. This is needed in the proof of the dual representation for the bid price of an American option.

A European option obliges the seller (writer) to deliver a portfolio $\zeta
\in\mathcal{L}_{T}$ at time~$T$. The set of strategies superhedging the
seller's position is given as%
\[
\Psi^{\mathrm{a}}(\zeta):=\left\{  y\in\Phi\st y_{t}-y_{t+1}\in\mathcal{K}%
_{t}\text{ for }t=0,\ldots,T-1,y_{T}-\zeta\in\mathcal{K}_{T}\right\}
\]
and the \emph{ask price} (\emph{seller's price}) of such an option in currency
$j=1,\ldots,d$ is%
\[
p_{\cc}^{\mathrm{a}}(\zeta):=\inf\left\{  x\in\mathbb{R}\st \exists y\in\Psi^{\mathrm{a}}(\zeta): xe^{\cc}%
=y_{0}\right\}  .
\]
The following result can be found in \citet[Section~4.3.1]{Roux_Zastawniak2011}.

\begin{lemma}\label{Lem:Eur-ask-price-dual-repr}
The ask price in currency $j=1,\ldots,d$ of a European option~$\zeta$ can be represented as%
\[
p_{\cc}^{\mathrm{a}}(\zeta)=\max_{\left(  \mathbb{Q},S\right)  \in
\Pbarc}\mathbb{E}_{\mathbb{Q}}(\zeta\cdot S_{T}).
\]
Moreover, a pair $(\hat{\mathbb{Q}},\hat{S})\in\Pbarc$ such that
$p_{\cc}^{\mathrm{a}}(\zeta)=\mathbb{E}_{\hat{\mathbb{Q}}}(\zeta\cdot \hat{S}_{T})$ can be
constructed algorithmically.
\end{lemma}


\end{document}